\documentclass[12pt,draftclsnofoot,onecolumn]{IEEEtran}
\usepackage{tabularx}

\usepackage{algpseudocode}
\usepackage{algorithm}
\usepackage{algorithmicx}
\usepackage{lipsum} 
\usepackage{arydshln}
\usepackage[dvips]{color}
\usepackage{comment}
\usepackage{todonotes}
\usepackage{epsf}
\usepackage{epsfig}
\usepackage{times}
\usepackage{epsfig}
\usepackage{graphicx}
\usepackage{bbold}
\usepackage{mathtools}
\usepackage{mathrsfs}
\usepackage{amssymb}
\usepackage{pdfpages}
\usepackage{epstopdf}
\newfloat{algorithm}{t}{lop}

\usepackage{amsmath}

\usepackage{dsfont}
\usepackage{lettrine}
\usepackage{amsmath,epsfig,amssymb,algorithm,algpseudocode,amsthm,cite,url}

\usepackage{subcaption}
\allowdisplaybreaks
\usepackage{csquotes}

\usepackage{verbatim}

\usepackage[english]{babel}

\usepackage{mathtools, cuted}
\usepackage{amsmath,amssymb}

\usepackage[justification=centering]{caption}
\usepackage{verbatim}

\newtheorem{corollary}{Corollary}
\newtheorem{theorem}{\bf Theorem}

\newtheorem{proposition}{\bf Proposition}
\newtheorem{lemma}{\bf Lemma}

\newtheorem{definition}{\bf Definition}
\newtheorem{remark}{Remark}

\newcommand{\splitatcommas}[1]{%
	\begingroup
	\ifnum\mathcode`,="8000
	\else
	\begingroup\lccode`~=`, \lowercase{\endgroup
		\edef~{\mathchar\the\mathcode`, \penalty0 \noexpand\hspace{0pt plus 1em}}%
	}\mathcode`,="8000
	\fi
	#1%
	\endgroup
}

\makeatletter
\newcommand\semihuge{\@setfontsize\semihuge{22.3}{22}}
\makeatother

\begin{document}
	%
	\title{{Colonel Blotto Game for Secure State Estimation in Interdependent Critical Infrastructure}	\vspace{-0.2cm}}
	\IEEEoverridecommandlockouts
	\author{\IEEEauthorblockN{Aidin Ferdowsi\IEEEauthorrefmark{1}, Walid Saad\IEEEauthorrefmark{1}, and Narayan B. Mandayam\IEEEauthorrefmark{2}\\}
		\IEEEauthorblockA{\IEEEauthorrefmark{1}
			Wireless@VT, Bradley Department of Electrical and Computer Engineering, \\ Virginia Tech, Blacksburg, VA, USA,
			Emails: \{aidin,walids\}@vt.edu\\}
		\IEEEauthorblockA{\IEEEauthorrefmark{2}
			WINLAB, Dept. of ECE, Rutgers University, New Brunswick, NJ, USA,		Email:  narayan@winlab.rutgers.edu}\vspace{-1cm}
		\thanks{This research was supported by the U.S. National Science Foundation under Grants OAC-1541105, OAC-1541069, and EAGER-1745829.\vspace{-0.5cm}}
	}
	\maketitle
	

	%
	\IEEEpeerreviewmaketitle
	
	\begin{abstract}
		Securing the physical components of a city's interdependent critical infrastructure (ICI) such as power, natural gas, and water systems is a challenging task due to their interdependence and a large number of involved sensors. In this paper, using a novel integrated state-space model that captures the interdependence, a two-stage cyber attack on an ICI is studied in which the attacker first compromises the ICI's sensors by decoding their messages, and, subsequently, it alters the compromised sensors' data to cause state estimation errors. To thwart such attacks, the administrator of each critical infrastructure (CI) must assign protection levels to the sensors based on their importance in the state estimation process. To capture the interdependence between the attacker and the ICI administrator's actions and analyze their interactions, a Colonel Blotto game framework is proposed. The mixed-strategy Nash equilibrium of this game is derived analytically. At this equilibrium, it is shown that the administrator can strategically randomize between the protection levels of the sensors to deceive the attacker. Simulation results coupled with theoretical analysis show that, using the proposed game, the administrator can reduce the state estimation error by at least $ 50\% $ compared to {a non-strategic approach that assigns protection levels proportional to sensor values}.	\end{abstract}	 \vspace{-4mm}
	\section{Introduction}\vspace{-1mm}
	The services delivered by a smart city's critical infrastructure (CI) such as power, natural gas, and water will be highly interdependent \cite{Chattopadhyay,farhangi2010path,rinaldi2001identifying, nam2011smart}. CIs are cyber-physical systems (CPSs) that encompass physical infrastructure whose performance is monitored and controlled by a cyber system, typically consisting of a massive number of sensors. These CPSs exhibit close interactions between their cyber and physical components\cite{nam2011smart,ferdowsi2017deep,ouyang2014review}. {The different types of interdependencies inside and between CPSs include: 1) Physical, in which a CPS's state depends on the output of another CPS, 2) Cyber, a CPS's state depends on the received information from another CPS, 3) Policy-related, where the administrative decisions impact the CPSs, 4) Shared, in which CPSs share some components, and 5) Exclusive, in which only one CPS can work at a time while other CPSs must wait for the CPS to finish operating. Thus, the various interdependencies within CPSs require having precise state estimation for monitoring purposes\cite{ouyang2014review,Rana2017,Kamel1996}.}\vspace{-3mm}
	
	\subsection{Previous Works}\vspace{-2mm}
	The state estimation of the CIs, which uses cyber elements to monitor the physical elements, is a crucial stage for controlling their functionality. However, the interdependency between CIs and the high synergy between their physical and cyber components make them vulnerable to attacks and failures \cite{Sinopoli2011integ,yagan2012optimal,ferdowsi2017game}. Numerous solutions have been presented for securing state estimation of CPSs as well as for CI failure detection and identification \cite{pasqualetti2013attack,fawzi2014secure,kwon2013security,mo2010false}. In \cite{pasqualetti2013attack}, the authors presented a control-theoretic approach for attack detection and identification in noiseless environments using centralized and distributed attack detection filters. The works in \cite{fawzi2014secure,kwon2013security,mo2010false} considered the estimation of a CPS under stealthy deception and replay cyber-attacks using a Kalman filter (KF). Moreover, the security of \emph{interdependent critical infrastructure (ICI)} has been studied in recent works such as \cite{ouyang2015resilience,nan2017quantitative,chang2014toward}. In \cite{ouyang2015resilience}, the authors assessed the security of interdependent power and natural gas CIs under multiple hazards, considering the ICI's performance as a measurement for security. In \cite{nan2017quantitative}, the authors proposed an agent-based model to capture the effects of interdependencies and quantify the coupling strength within ICIs. Also, the impact of natural and human-included disasters has been studied in \cite{chang2014toward}. Furthermore, the security and protection of sensor networks, which collect data from CIs has been studied in \cite{szefer2012physical,ashok2014cyber,ferdowsi2018deep}. In \cite{szefer2012physical}, the authors proposed a novel method for physical attack protection with human virtualization in the context of data centers using sensors that detect an impending physical/human attack and, then, alarm to mitigate the attack. The work in \cite{ashok2014cyber} proposed a distributed observer for state estimation of CIs in lossy sensor networks with cyber attacks. The authors in \cite{ferdowsi2018deep} proposed a deep learning algorithm to authenticate vulnerable sensors in an Internet of Things network.  
	
	The works in {\cite{gupta2014three,ferdowsi2017colonel,schwartz2014heterogeneous,GCBG2018,kovenock2015generalizations}} used a Colonel Blotto game (CBG) to study the interactions between a CPS's defender and an attacker that seeks to compromise the CPS components. The CBG captures the competitive interactions between two players that seek to allocate resources across a set of battlefields. The player who allocates more resources to a certain battlefield wins it and receives a corresponding valuation. In \cite{gupta2014three}, a three stage CBG has been proposed to analyze the interaction of an attacker with two defenders. The work in \cite{ferdowsi2017colonel} studied the resilience of smart cities against cyber attacks using a CBG framework. In addition, many variants of the CBG have been studied including those	with symmetric resources \cite{kovenock2015generalizations}, heterogeneous resources \cite{schwartz2014heterogeneous}, and approximate winning-losing setting \cite{GCBG2018}.
	
	However, the works in {\cite{pasqualetti2013attack,fawzi2014secure,kwon2013security,mo2010false,ouyang2015resilience,nan2017quantitative,chang2014toward,szefer2012physical,ashok2014cyber,ferdowsi2018deep}} do not consider the limitations of the available security resources for the protection, detection, and identification of CI attacks. For instance, in practical smart cities, resource limitations may substantially affect the security of the CIs. Indeed, because of massive data transmission from sensors to the central processing unit, security solutions such as in \cite{pasqualetti2013attack}, \cite{fawzi2014secure}, and \cite{ferdowsi2018deep} will require a large number of computations, a high communication bandwidth, a large amount of power, and
	a considerable level of financial resources, all of which constitute limited resources for the ICI's administrator. Therefore, unlike the idealized security solutions in {\cite{pasqualetti2013attack,fawzi2014secure,kwon2013security,mo2010false,ouyang2015resilience,nan2017quantitative,chang2014toward,szefer2012physical,ashok2014cyber,ferdowsi2018deep}}, due to resource limitations, the administrator of an ICI has to prioritize between the protection of the cyber components of ICI based on their importance in the state estimation process \cite{Rullo2017}. Another key limitation in the current literature is that the majority of the existing works, such as \cite{pasqualetti2013attack,fawzi2014secure,kwon2013security,mo2010false} and {\cite{gupta2014three,ferdowsi2017colonel,schwartz2014heterogeneous,kovenock2015generalizations,GCBG2018}} do not take into account the interdependence between the CPSs. Meanwhile, those that account for interdependencies such as in \cite{ouyang2015resilience,nan2017quantitative,chang2014toward} are mostly based on graph-theoretic constructs that abstract much of the functionalities of the CIs. In practice, the CIs are interdependent and cannot be simply captured by a graph. 
\vspace{-3mm}
\subsection{Contributions}\vspace{-2mm}
	The main contribution of this paper is a novel game-theoretic framework for analyzing and optimizing the security of a large-scale ICI's state estimation. To build this unified security framework, this paper makes several contributions:{
	\vspace{-0.1cm}
	\begin{itemize}
		\item We first introduce a novel integrated state-space model that captures the dynamics of an ICI consisting of power, natural gas, and water distribution systems. {To the best of our knowledge, this is the first model that mathematically captures the interdependence of these three CIs.} For enabling state estimation of the proposed ICI dynamics model, we implement a centralized KF that uses the sensor data to estimate the ICI's state. 
		\item {We derive the maximum state estimation error deviation on the ICI caused by} a two-stage cyber attack that targets the sensors of the ICI so as to manipulate the state estimation. {In essence, in the first stage, the attacker aims to compromise the ICI sensors} by breaking their protection algorithm (e.g., watermarking or sensor attack detection filter). { We address this attack stage by assigning protection levels on the sensors which are derived from a game-theoretic analysis.} In the second stage, the attacker manipulates the ICI's state estimation by altering the compromised sensors' data to induce state estimation errors. {To defend against the second stage of the attack and protect the sensors, we implement an attack detection filter based on a Kullback-Leibler (KL) divergence. Using the KL divergence we can derive a maximum cumulative state estimation error deviation caused by manipulating any sensor in the ICI. This is a notable result since it enables the defender to distinguish most valuable sensors and protect the ICI accordingly. }
		\item Since the actions of the attacker and the defender are interdependent, we propose a \emph{Colonel Blotto game} framework\cite{Roberson2006} to analyze the interactions between the attacker and the administrator. In this game, the attacker chooses the set of sensors to compromise while the administrator assigns protection levels to the sensors. In contrast to existing works on Colonel Blotto for CPS security {\cite{gupta2014three,ferdowsi2017colonel,schwartz2014heterogeneous,GCBG2018,kovenock2015generalizations}}, our game considers the interdependence between multiple CPSs. For this game, we derive the mixed-strategy Nash equilibrium for the administrator and the attacker as a function of their available resources and the maximum state estimation error due to the attack.
	\end{itemize}}

Extensive simulations are used to corroborate the theoretical findings. Simulation results show that the administrator's mixed strategy increases the security of large-scale ICIs and reduces the state estimation error of the ICI by at least a factor of $ 50\% $ compared to a baseline.\vspace{-3mm}
\section{Interdependent Critical Infrastructure and Attack Model} \label{CImodel}\vspace{-2mm}
	Consider an ICI as a CPS whose \emph{physical system} consists of three interdependent power, natural gas, and water distribution CIs and whose \emph{cyber system} is a network of sensors that collect data from the physical components of the CIs and transmit it to a central processing unit. We first derive a state-space model for the physical system of each CI separately and then present the general ICI model. {The detailed derivations can be found in \textbf{Appendix \ref{ICIAppendix}}.} Finally, we discuss the associated cyber system and its vulnerability to attacks.\vspace{-3mm}
	\subsection{Physical System}\vspace{-2mm}
	The power system can be modeled as a linear dynamic system whose inputs are the electrical power demands from the load buses \cite{andersson2012dynamics}. We focus only on generators that are supplied by natural gas \cite{lubega2014quantitative} and we consider water as a requirement for the vapor condensation and cooling down in some of the generators\cite{maffezzoni1997boiler}. The natural gas and water CIs are designed to supply natural gas and water to consumers in a city. Due to the pressure loss at the junctions of these two CIs, gas compressors \cite{alamian2012state} and water pumps\cite{Burgschweiger2009} are used to compensate the pressure loss. { Fig. \ref{fig:ICIexample} shows a block diagram of such an ICI model.}
	Therefore, we can write a state-space model for the interdependent critical gas-power-water infrastructure:
	\begin{align}\label{eq:interCI}
	\dot{\boldsymbol{x}}(t)={\bar{\boldsymbol{A}}}\boldsymbol{x}+{\bar{\boldsymbol{B}}}\boldsymbol{u}(t),
	\end{align}
	{where $ \boldsymbol{x}(t) \in \mathbb{R}^{n} $ is an $ {n\times 1} $ vector and $ \boldsymbol{u}(t) \in \mathbb{R}^{\tilde{n}} $ is an $ {\tilde{n}\times 1} $ vector. $ \boldsymbol{x}(t)$ is the vector that contains the all state variables such as the power flows in the power CI and pipe pressure in water and natural gas CI and $ \boldsymbol{u}(t) $ contains the power, gas, and water demands from the end-users. Moreover, $ {\bar{\boldsymbol{A}}} \in \mathbb{R}^{n\times n} $ is an $ {n\times n} $ matrix and $ {\bar{\boldsymbol{B}}} \in \mathbb{R}^{n\times \tilde{n}} $ is an $ {n\times \tilde{n}} $. $ {\bar{\boldsymbol{A}}} $ and they are the representation of the interdependencies between the ICI's state variables. In addition, $ n $ is the total number of states in the ICI and $ \tilde{n} $ is the total number of power demands from the ICI.} {For ease of exposition, an extensive explanation of the interdependence and our derivation of all the state space model matrices for the interdependence of CIs \emph{\textbf{are summarized in Appendix \ref{ICIAppendix}}}}.

\begin{figure}[!t]
	\centering
	\captionsetup{singlelinecheck = false, justification=justified}
	\includegraphics[width=0.45\columnwidth]{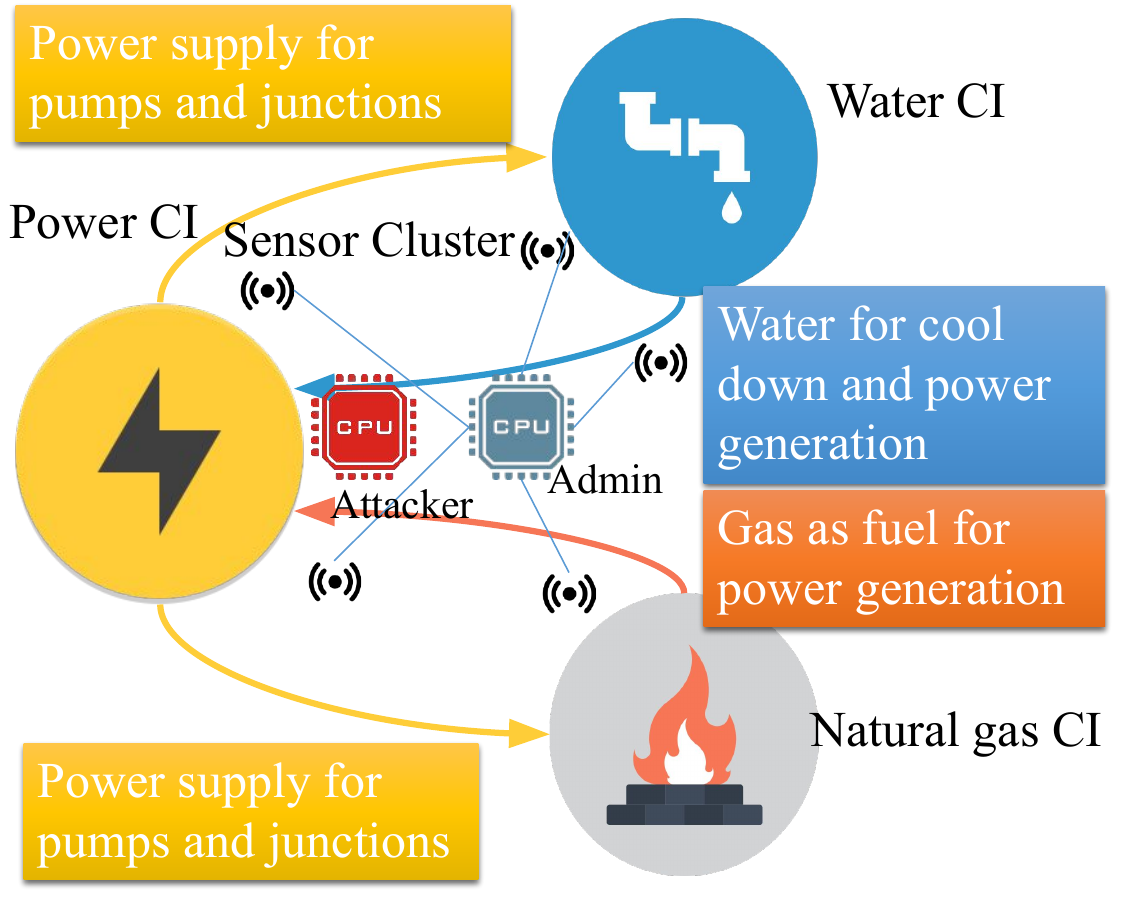}
	\vspace{-0.2cm}
	\caption{{An illustrative example of an ICI.}}
	\label{fig:ICIexample}
	\vspace{-0.7cm}
\end{figure}

\vspace{-3mm}
\subsection{Cyber System}\vspace{-2mm}
To monitor the state variables in \eqref{eq:interCI}, a cyber system is needed. For the considered ICI , the cyber system will consist of a number of sensors spread around the ICI and collecting different measurements from the ICI's components. Sensors and meters in the power infrastructure measure the instantaneous frequency of the generator, the mechanical input power to the generator, and the line powers between the generators. In the natural gas and water CI, sensors collect the outlet pressure, and inlet flow rate of each pipeline. As shown in Fig. \ref{fig:ICIexample}, we consider a sensor network that is used to collect data from the ICI and send it to a central server. The sensor data collected from each CI can be expressed as a linear equation of the states of the ICI, as follows:
	\begin{align}\label{eq:sensor}
	\boldsymbol{y}(t)=\bar{\boldsymbol{C}}\boldsymbol{x}(t),
	\end{align}  
	where $ \boldsymbol{x}(t) $ is given in \eqref{eq:interCI}, and $ \boldsymbol{y}(t) \in \mathbb{R}^{p\times 1} $ is a $ p\times 1 $ vector of all the sensor data at each time instant,  $\bar{\boldsymbol{C}} \in \mathbb{R}^{p\times n} $ is a $ {p\times n} $ matrix for converting the states to the sensor data and $ p $ is the total number of sensors in ICI. However, due to the inaccuracy in measurements and the process noise in the infrastructure as well as possibility of $ {\bar{\boldsymbol{C}}} $ not being full rank, the owner of each CI must estimate the system state at each time instant. Due to the interdependence between the CIs, their owners have to share the collected data from the components with a single administrator who has access to the ICI model \cite{zanella2014internet}. Note that a lack of cooperation between the owners of the CIs can yield estimation error since the administrator will not be able to capture the interdependencies.  {While \eqref{eq:interCI} and \eqref{eq:sensor} capture the physical and cyber behavior of the ICI, however, they do not consider the process and measurement noise, and also the discrete sensor data. Thus, here, we transform the continuous state-space model equations and the sensor outputs to a discrete model using a bilinear transformation \cite{williams2007linear}}:
	\begin{equation}\label{discreteICI}
	\begin{aligned}
	\boldsymbol{x}(k+1)&=\boldsymbol{A}\boldsymbol{x}(k)+\boldsymbol{B}\boldsymbol{u}(k)+\boldsymbol{w}(k),\\
	\boldsymbol{y}(k)&=\boldsymbol{C}\boldsymbol{x}(k)+\boldsymbol{l}(k),
	\end{aligned}
	\end{equation}
	{where $ \boldsymbol{A} \in \mathbb{R}^{n\times n} $, $ \boldsymbol{B} \in\mathbb{R}^{n\times \tilde{n}} $, and $ \boldsymbol{C}\in\mathbb{R}^{p\times n} $ are $ {n\times n} $, $ {n\times \tilde{n}} $, and $ {p\times n} $ matrices and are discretized versions of the matrices} defined in \eqref{eq:interCI} and \eqref{eq:sensor}, $ \boldsymbol{x}(k) \in \mathbb{R}^{n} $ is the $ {n\times 1} $ vector of state variables of the ICI at time step $ k $, $ \boldsymbol{u}(k) \in \mathbb{R}^{p} $ is the $ {\tilde{n}\times 1} $ vector of external inputs of the ICI at time step $ k $, $ \boldsymbol{w}(k) \in \mathbb{R}^{n} $ is the $ {n\times 1} $ vector of process noise at time $ k $, and $ \boldsymbol{l}(k) \in \mathbb{R}^{p} $ is the $ p\times 1 $ vector of measurement noise at time $ k $. Due to the discrete sensor data, hereinafter, we use \eqref{discreteICI} in our analysis which is the discrete model for the ICI. {\emph{Note that, we transform the $ (\bar{\boldsymbol{A}},\bar{\boldsymbol{B}}, \bar{\boldsymbol{C}}) $ matrices into the discrete form  $ (\boldsymbol{A},\boldsymbol{B}, \boldsymbol{C}) $ using a bilinear transformation\cite{williams2007linear}.}} In addition, $ \boldsymbol{x}(0) $ is the initial state of the ICI, and $ \boldsymbol{w}(k) $, $ \boldsymbol{l}(k) $ and $ x(0) $ are independent Gaussian random variables with $ x(0) \sim \mathcal{N}(\boldsymbol{0},\boldsymbol{\Psi})$, $ \boldsymbol{w}(k) \sim \mathcal{N}(\boldsymbol{0},\boldsymbol{\Phi})$, and $ \boldsymbol{l}(k) \sim \mathcal{N}(\boldsymbol{0},\boldsymbol{\Omega})$ { where $ \boldsymbol{\Psi} \in \mathbb{R}^{n\times n} $, $ \boldsymbol{\Phi} \in \mathbb{R}^{n\times n} $, and $ \boldsymbol{\Omega} \in \mathbb{R}^{p\times p} $ are ${n\times n}$, ${n\times n} $, and $ {p\times p} $ matrices, respectively.}
	
	The ICI administrator seeks to estimate the state of the ICI using \eqref{discreteICI}. However, due to sensor error and operation noise, a noise-resilient method is needed to estimate the state variables. To this end, it can be shown that by using a KF, one can compute the state estimation $ \hat{\boldsymbol{x}}(k) $ from observations $ \boldsymbol{y}(k) $\cite{kumar2015stochastic}. Since the initial time of the ICI is considered $ -\infty $, the KF converges to a fixed gain linear estimator. To find the state estimate of the system, we first compute the $ {n\times n} $ Kalman state probability matrix $ \boldsymbol{P} \in \mathbb{R}^{n\times n} $ as $
	\boldsymbol{P}= \boldsymbol{A}\boldsymbol{P}\boldsymbol{A}^T\hspace{-0.1cm}+\boldsymbol{\Phi}-\boldsymbol{A}\boldsymbol{P}\boldsymbol{C}^T\hspace{-0.1cm}\left(\boldsymbol{C}\boldsymbol{P}\boldsymbol{C}^T+\boldsymbol{\Omega}\right)^{-1}\hspace{-0.3cm}\boldsymbol{C}\boldsymbol{P}\boldsymbol{A}^T.$
	Then, we compute the $ {n\times p} $ Kalman fixed gain matrix as follows $
	\boldsymbol{K} =\boldsymbol{P}\boldsymbol{C}^T\left(\boldsymbol{C}\boldsymbol{P}\boldsymbol{C}^T+\boldsymbol{\Omega}\right)^{-1}.$
	Next, we find the state estimation vector at time $ k $ dependent to knowing the state estimation at time $ k-1 $, $ \hat{\boldsymbol{x}}(k|k-1) $ as $
	\hat{\boldsymbol{x}}(k|k-1)=\boldsymbol{A}\hat{\boldsymbol{x}}(k-1|k-1)+\boldsymbol{B}\boldsymbol{u}(k-1).$
	Finally, we compute the state estimate $ \hat{\boldsymbol{x}}(k|k) $ using a KF:
	\begin{align}\label{estimate}
	\hat{\boldsymbol{x}}(k|k)=(\boldsymbol{I}-\boldsymbol{K}\boldsymbol{C})\hat{\boldsymbol{x}}(k|k-1)+\boldsymbol{K}\boldsymbol{y}(k),
	\end{align}
	where the initial state is defined as $ \hat{\boldsymbol{x}}(0|0)=\boldsymbol{x}(0) $. {The initial state of the KF for a linear system only affects the convergence rate but the final estimation will not change. Therefore, irrespective of how the initial state is selected, the KF can converge to the optimal estimation \cite{williams2007linear}. In our case, we initialize the ICI state based on our knowledge about the steady state operation. Another commonly used initialization method for a KF is assigning zero values to all of the states. We define the \emph{estimation error} as the difference between the state $ \boldsymbol{x}(k) $ and its estimate $ \hat{\boldsymbol{x}}(k) $:
	\begin{equation}\label{errdef}
	\boldsymbol{e}(k) \triangleq \boldsymbol{x}(k) - \hat{\boldsymbol{x}}(k),
	\end{equation}
	{where $ \boldsymbol{e}(k) \in \mathbb{R}^{n} $ is an $ {n\times 1} $ vector. }
	Using \eqref{estimate} and \eqref{errdef}, we have:
	\begin{align}
	\boldsymbol{e}(k+1)&=(\boldsymbol{A}-\boldsymbol{K}\boldsymbol{C}\boldsymbol{A})\boldsymbol{e}(k)+(\boldsymbol{I}-\boldsymbol{K}\boldsymbol{C})\boldsymbol{w}(k)
	-\boldsymbol{K}\boldsymbol{l}(k)\nonumber\\
	&+\boldsymbol{K}\boldsymbol{C}\boldsymbol{B}\left(\boldsymbol{u}(k)-(\boldsymbol{A}-\boldsymbol{K}\boldsymbol{C}\boldsymbol{A})\boldsymbol{u}(k-1)\right),
	\end{align}
	We also define the residue of the KF:
	$
	\boldsymbol{z}(k)\triangleq\boldsymbol{y}(k)-\boldsymbol{C}\boldsymbol{A}\hat{\boldsymbol{x}}(k),
	$
	{where $ \boldsymbol{z}(k) \in \mathbb{R}^{p\times 1} $ is a $ {p\times 1} $ vector.} Because of process and measurement noise, we need to validate the estimation of the states and detect the failure of the estimation filter. We use $ \mathcal{X}^2 $ failure detector allowing the detector computes the following value at each time step \cite{chen2012robust}:
	\begin{align}\label{failuredetect}
	g(k)= \boldsymbol{z}^T(k)\boldsymbol{\mathcal{Z}}^{-1}\boldsymbol{z}(k),
	\end{align}
	where $ \boldsymbol{\mathcal{Z}} { \in \mathbb{R}^{p \times p}} $ {is a $ p\times p $ semi-positive definite matrix and is the relative cost of residue vector $ \boldsymbol{z}(k) $}. If $ g(k) $ exceeds the threshold level, then the detector will trigger an alarm.\vspace{-4mm}
	\subsection{Attack Model}
	Consider the cyber system of the ICI in Fig. \ref{fig:ICIexample}, where sensors collect measurements from the physical components of the ICI and transmit the measurement data to a \emph{central node} in their proximity. Then different central nodes will transmit the data to a central \emph{server} that will calculate the estimation of ICI state variables using the presented KF in \eqref{estimate}. We refer to the group of sensors which connect to a single central node, as a \emph{sensor cluster} (SC). We consider a two-stage attack model to the cyber system of our ICI. In the first stage, the attacker aims to compromise the ICI's sensors by breaking the security solution that is implemented by the ICI administrator (referred to as the \emph{defender}, hereinafter). After compromising some of the SCs, in the second stage, the attacker manipulates the SC data to increase the ICI state estimation error. 
	
	Our model can be used to capture any ICI security solution that can include a watermarking of the sensor data\cite{ferdowsi2018deep}, an implementation of attack detection filter\cite{pasqualetti2013attack}, or a physical protection of the sensors\cite{szefer2012physical}. Therefore, to compromise any SC within the ICI, the attacker has to collect the broadcast data from the sensors to the central nodes and compromise the implemented security solution. However, this requires processing of the collected data from the sensors across the ICI, physical presence of the attacker in the proximity of the central nodes to collect data, or communication resources for transmission of the collected data to the attacker's central processing unit. Since processing, communication and human resources are limited, the attacker needs to prioritize between the sensors based on their importance in the state estimation of the ICI. From the defender's point of view, implementing the aforementioned security solutions, requires computational resources, communication bandwidth, or financial resources which are restricted in availability for the defender. Therefore, the defender must also prioritize between the ICI sensors that it seeks to protect. Thus, to initiate this two-stage attack, the attacker must have information about the implemented security protocols at SCs, the defender's available security resources, and the ICI cyber-physical model. This attack model has been widely used in the literature \cite{pasqualetti2013attack,fawzi2014secure,kwon2013security} and is appropriate to consider because using this assumption we secure the ICI against the most capable attacker and, thus, for any other attackers with less capabilities, the ICI will still be secure. The defender's required information is the available resources of the attacker, i.e., how strong is the attacker. The defender can always assume a worst-case or average value for the attacker's available resources based on typical attackers’ capabilities, past attacks, or known data on similar attacks\cite{pasqualetti2013attack,fawzi2014secure,kwon2013security}.
	
	In summary, the attacker aims to maximize the state estimation error through the compromised sensors and the defender seeks to protect the SCs of the ICI from this cyber attack, under strict resource limitations at both sides. To analyze this interactions between the attacker and the defender, first, we study the second stage of attack to find the maximum estimation error caused by the cyber attack and quantify the importance of each SC in the ICI, then using these values we can formally analyze the attacker-defender interaction and derive optimal defense strategies.\vspace{-3mm}
	\section{Maximum State Estimation Error in the Compromised Sensors}\label{CED}\vspace{-2mm}
	In this section, we analyze the impact of the second stage of the cyber attack in order to quantify the ability of an attacker to increase the estimation error by altering the sensor data. We assume a worst-case scenario for security analysis in which the attacker has complete knowledge about the system as done in \cite{pasqualetti2013attack} and was able to compromise some of the SCs in the first stage. We assume that the attacker can change the data of the compromised sensors to a desired value in order to disturb the ICI's state estimation. Given the set of all compromised SCs,  $ \mathcal{A} $, we define attack vector at time step $ k $, $ \boldsymbol{y}^a(k) \triangleq [{\boldsymbol{y}_1^a}^T(k) ,\dots, {\boldsymbol{y}_N^a}^T(k)]^T $ where $ N $ is the number of SCs, and $ {\boldsymbol{y}_i^a} $ is the $1\times N_i$ attack vector on SC $ i $ where $ N_i $ is the number of sensors in SC $ i $. Also, $ \boldsymbol{y}_i(k) = \boldsymbol{0} $ if $ i \notin \mathcal{A} $. Therefore, the linear relationship between the state variables of the ICI and the sensor data under attack will be $
	\boldsymbol{\bar{y}}(k)=\boldsymbol{C}\boldsymbol{x}(k)+\boldsymbol{l}(k)+\boldsymbol{y}^a(k),$
	where $ \boldsymbol{\bar{y}}(k) $ is the vector of sensor measurements under attack, and $ \boldsymbol{y}^a(k) $ is independent from $ \boldsymbol{w}(k) $, $ \boldsymbol{v}(k) $, and  $ \boldsymbol{x}(0) $. Here, we assume that the attack on the sensors starts from $ k=1 $. When the ICI's cyber system is under attack, the Kalman state estimation filter of the ICI in \eqref{estimate} changes as follows:
	\begin{align}\label{estimateunderattack}
	\bar{\boldsymbol{x}}(k|k-1)=\boldsymbol{A}\bar{x}(k-1|k-1)+\boldsymbol{B}\boldsymbol{u}(k-1),\,\,
	\bar{\boldsymbol{x}}(k|k)=(\boldsymbol{I}-\boldsymbol{K}\boldsymbol{C})\bar{\boldsymbol{x}}(k|k-1)+\boldsymbol{K}\bar{\boldsymbol{y}}(k),
	\end{align}
	where $ \bar{\boldsymbol{x}}(k|k) $ is the estimate of the states under attack. The new residue and estimation error are defined as $
	\bar{\boldsymbol{z}}(k)\triangleq\bar{\boldsymbol{y}}(k)-\boldsymbol{C}\boldsymbol{A}\bar{\boldsymbol{x}}(k-1)$ and 
	$\bar{\boldsymbol{e}}(k)\triangleq \boldsymbol{x}(k)-\bar{\boldsymbol{x}}(k).$
	We define the error difference at the ICI state estimation between the under attack and in absence of attack as $
	\Delta \boldsymbol{e}(k)\triangleq\bar{\boldsymbol{e}}(k) -\boldsymbol{e}(k),
	\Delta \boldsymbol{z}(k) \triangleq\bar{\boldsymbol{z}}(k)-\boldsymbol{z}(k).$
	Using \eqref{estimate} and \eqref{estimateunderattack}, we can find the following model for the difference in error and residue:
	\begin{align}
	\label{errdeviation}
	\Delta \boldsymbol{e}(k+1)&=(\boldsymbol{A}-\boldsymbol{K}\boldsymbol{C}\boldsymbol{A})\Delta \boldsymbol{e}(k) - \boldsymbol{K} \boldsymbol{y}^a(k+1), \\
	\label{resdeviation}
	\Delta \boldsymbol{z}(k+1)&=\boldsymbol{C}\boldsymbol{A}\Delta \boldsymbol{e}(k)+ \boldsymbol{y}^a(k+1).
	\end{align}
	We define the \emph{cumulative error difference} (CED) at time step $ k $:
	\begin{align}\label{cumerr}
	q(k) \triangleq \Delta \boldsymbol{e}^T(k)\boldsymbol{E} \Delta \boldsymbol{e}(k),
	\end{align}
	where $ \boldsymbol{E} \in \mathbb{R}^{n \times n} $ is the $n\times n $ relative cost matrix of state error. {$ \boldsymbol{E} $ is a positive semi-definite matrix such that $  \boldsymbol{Q}^T\boldsymbol{E}\boldsymbol{Q} \preceq \boldsymbol{E} $, i.e., for any vector $ \boldsymbol{v} $, we have   $ \boldsymbol{v}^T\boldsymbol{Q}^T\boldsymbol{E}\boldsymbol{Q}\boldsymbol{v} \leq  \boldsymbol{v}^T \boldsymbol{E}\boldsymbol{v} $ }, where $ \boldsymbol{Q}\triangleq\boldsymbol{A}-\boldsymbol{K}\boldsymbol{C}\boldsymbol{A} $ which is a Hurwitz stable matrix since the ICI model is a stable system \cite{kumar2015stochastic}. Next, we derive the maximum CED caused by an impulse attack (an attack vector that has nonzero values in the initial time step and zero values afterwards) to an SC. The reason for analyzing the impulse attack is that, any attack vector can be designed by combination of shifted impulse attack vectors. {This is because of the superposition characteristics of linear time invariant systems, the effect of any attack sequence on the ICI is equal to the summation of shifted impulse attack vectors \cite{williams2007linear}.} \vspace{-3mm}
	\begin{proposition}\label{proposition0}
		The maximum CED caused by an impulse attack to a set of sensors $ \mathcal{A}$ is:
		{\begin{equation}
		\begin{aligned}\label{lemma1eq}
		q^m(\boldsymbol{y}^a)\triangleq\boldsymbol{y}^{a^T}(1)\boldsymbol{K}^T\boldsymbol{E}\boldsymbol{K}\boldsymbol{y}^a(1).
		\end{aligned}
		\end{equation}}\vspace{-4mm}
	\end{proposition} \vspace{-4mm}
	\begin{proof}\vspace{-2mm}
		Since $ \boldsymbol{y}^a(1)$ is a vector with $ \boldsymbol{y}^a(1)[i]=0 $ for $ i \notin \mathcal{A} $ and the attack is an impulse input, then, we have $ \boldsymbol{y}^a(k)=0 $ for $ k>1 $. Therefore, using \eqref{errdeviation}, we have $ \Delta \boldsymbol{e}(k)= \boldsymbol{Q}\Delta \boldsymbol{e}(k-1) =  \boldsymbol{Q}^{(k-1)}\boldsymbol{K}\boldsymbol{y}^a(1)$. 
		Therefore, for $ k>2 $, we will have:\vspace{-1mm}
		\begin{align*}
		q(k)= \Delta \boldsymbol{e}^T(k) \boldsymbol{E} \Delta \boldsymbol{e}(k)  = \Delta \boldsymbol{e}^T(k-1) \boldsymbol{Q}^T \boldsymbol{E}  \boldsymbol{Q} \boldsymbol{e}(k-1)\preceq \Delta \boldsymbol{e}^T(k-1) \boldsymbol{E}  \boldsymbol{e}(k-1) = q(k-1),
		\end{align*}
		which means that the maximum value of $ q(k) $ occurs in $ q(1) = \boldsymbol{y}^{a^T}(1)\boldsymbol{K}^T\boldsymbol{E}\boldsymbol{K}\boldsymbol{y}^a(1) $.
	\end{proof}\vspace{-3mm}
	Proposition \ref{proposition0} shows that the maximum CED caused by an impulse attack occurs during first time instant, $ q(1) $, after the initiation of the impulse attack to the sensors. In the presence of the attack, the failure detector in \eqref{failuredetect} computes the following value in each time step $
	\bar{g}(k)=\bar{\boldsymbol{z}}^T(k)\boldsymbol{\mathcal{Z}}^{-1}\bar{\boldsymbol{z}}(k)$.
	Now, we define two new parameters for the analysis of probability of failure in the system as follows: $
	\beta(k) \triangleq \textrm{Pr}(g(k)>g^{t})$ and $ 
	\bar{\beta}(k)  \triangleq \textrm{Pr}(\bar{g}(k)>g^{t}),
	$
	where $ \beta(k) $  and $ \bar{\beta}(k) $ capture probabilities of failure in absence and existence of attack, respectively, and $ g^{t} $ is the failure trigger threshold. {Also, for the subsequent analysis, we consider that $ \boldsymbol{\mathcal{Z}} $ is chosen such that $ \boldsymbol{A}^T\boldsymbol{C}^T \boldsymbol{\mathcal{Z}}^{-1} \boldsymbol{C}\boldsymbol{A} \succeq\boldsymbol{Q}^T \boldsymbol{A}^T\boldsymbol{C}^T \boldsymbol{\mathcal{Z}}^{-1} \boldsymbol{C}\boldsymbol{A}\boldsymbol{Q}$.}
	\begin{definition}
	An impulse attack to set $ \mathcal{A}$ is \emph{$ \alpha $-feasible} if:
	\begin{align}
		D(\boldsymbol{z}(k)||\bar{\boldsymbol{z}}(k))=||\Delta \boldsymbol{z}(k)||_{\boldsymbol{S}}=\sqrt{\Delta\boldsymbol{z}^T(k)\boldsymbol{S}\Delta\boldsymbol{z}(k)}\leq \alpha.
	\end{align}
	for all $ k=1, \dots , \infty $, { where $ \boldsymbol{S} \in \mathbb{R}^{p\times p} $ is a $ p\times p $ matrix defined as } $ \boldsymbol{S} \triangleq\mathcal{Z}^{-1}/2 $ and $ D(\boldsymbol{z}(k)||\bar{\boldsymbol{z}}(k)) $ is the KL distance between $ \boldsymbol{z}(k) $ and $ \bar{\boldsymbol{z}}(k)$.
	\end{definition}
	Using \cite[Theorem 1]{mo2010false}, we can directly prove the convergence of $ \bar{\beta}(k) $ to $ \beta(k) $ as $ D(\boldsymbol{z}(k)||\bar{\boldsymbol{z}}(k))$ goes to $ 0$, as follows.
	\begin{lemma}\label{lemma2}
		For any $ \epsilon > 0 $, there exists $ \alpha >0 $, such that if $
		D(\boldsymbol{z}(k)||\bar{\boldsymbol{z}}(k)) \leq \alpha,
		$
		for $ k=1,\dots, \infty $, then 
		$
		\bar{\beta}(k)\leq \beta(k)+\epsilon
		$
		for all $ k=1,\dots, \infty $.
	\end{lemma}
	Lemma \ref{lemma2} shows that, if the probability of alarm triggering at time $ k $, $ \beta(k) $, increases by a value of $ \epsilon $ in presence of attack, $ \bar{\beta}(k)=\beta (k)+\epsilon $, then, there exists a value for $ \alpha $ such that an impulse attack can be designed with a KL distance lower than $ \alpha $. {In other words, if the defender wants to increase the probability of triggering an alarm by, e.g., increasing the alarm threshold, then the KL divergence will still be bounded by $ \alpha $. We use the same concept to define the $ \alpha $-feasible attack. Essentially, we assume that, if the defender wants to keep the KL divergence under $ \alpha $, then the probability of triggering the attack will have a very small variation.} Now, if the attacker wants to design an $ \alpha $-feasible impulse attack then it should change the sensor data such that the KL distance never exceeds $ \alpha $. Next, we find the maximum KL distance caused by an impulse attack to a set of sensors. 
	\begin{lemma}\label{lemma3}
		The maximum KL divergence caused by an impulse attack to a set of sensors $ \mathcal{A}$ is:{
		\begin{align}
		D^m(\boldsymbol{y}^a)\triangleq\max\Bigg\{
		\sqrt{\left(\boldsymbol{C}\boldsymbol{A}\boldsymbol{K}\boldsymbol{y}^{a}(1)\right)^T\boldsymbol{S}\boldsymbol{C}\boldsymbol{A}\boldsymbol{K}\boldsymbol{y}^{a}(1)},\sqrt{\boldsymbol{y}^{a^T}(1)\boldsymbol{S}\boldsymbol{y}^{a}(1)}\Bigg\}.\label{lemma3eq}
		\end{align}}
	\end{lemma}   
	\begin{proof}
		From \eqref{errdeviation} and \eqref{resdeviation} we have $
		\Delta \boldsymbol{z}(1)=\boldsymbol{y}^a(1)$ and $
		\Delta \boldsymbol{z}(k)=-\boldsymbol{C}\boldsymbol{A}\boldsymbol{Q}^{(k-2)}\boldsymbol{K}\boldsymbol{y}^{a}(1), \quad k>2,$ therefore, the KL divergence for $ k>2 $ will be:{
		\begin{align}
			D(\boldsymbol{z}(k)||\bar{\boldsymbol{z}}(k))&=\sqrt{\Delta\boldsymbol{z}^T(k)\boldsymbol{S}\Delta\boldsymbol{z}(k)}=\sqrt{\Delta\boldsymbol{z}^T(k)\boldsymbol{\mathcal{Z}}^{-1}/2\Delta\boldsymbol{z}(k)}\nonumber\\&= \sqrt{\boldsymbol{y}^{a^T}(1)\boldsymbol{K}^T\boldsymbol{Q}^{{(k-2)}^T}\boldsymbol{A}^T\boldsymbol{C}^T\boldsymbol{\mathcal{Z}}^{-1}/2\boldsymbol{C}\boldsymbol{A}\boldsymbol{Q}^{(k-2)}\boldsymbol{K}\boldsymbol{y}^{a}(1)}\nonumber\\&
			=\sqrt{\boldsymbol{y}^{a^T}(1)\boldsymbol{K}^T\boldsymbol{Q}^{{(k-3)}^T}\boldsymbol{Q}^T\boldsymbol{A}^T\boldsymbol{C}^T\boldsymbol{\mathcal{Z}}^{-1}/2\boldsymbol{C}\boldsymbol{A}\boldsymbol{Q}\boldsymbol{Q}^{(k-3)}\boldsymbol{K}\boldsymbol{y}^{a}(1)}\nonumber\\&
			\leq \sqrt{\boldsymbol{y}^{a^T}(1)\boldsymbol{K}^T\boldsymbol{Q}^{{(k-3)}^T}\boldsymbol{A}^T\boldsymbol{C}^T\boldsymbol{\mathcal{Z}}^{-1}/2\boldsymbol{C}\boldsymbol{A}\boldsymbol{Q}^{(k-3)}\boldsymbol{K}\boldsymbol{y}^{a}(1)}\nonumber\\&
			 = D(\boldsymbol{z}(k-1)||\bar{\boldsymbol{z}}(k-1)).\label{disineq}
		\end{align}}
		\eqref{disineq} implies that the KL distance is decreasing for $ k>2 $, and hence, the maximum KL distance will occur in {$ k=1$ or $2 $} and this proves \eqref{lemma3eq}.
	\end{proof}
	Lemma \ref{lemma3} finds the maximum KL divergence caused by an impulse attack to a set of sensors. We use the maximum error caused by an impulse attack and maximum KL distance to find the maximum CED caused by an $ \alpha $-feasible attack in the following theorem. This theorem quantifies the maximum CED that the attacker can cause without triggering the alarm to a set of sensors.\vspace{-2mm}
\begin{theorem}\label{Theorem1}\vspace{-2mm}
For any value of $\alpha$ chosen by the defender, the maximum CED caused by an impulse $ \alpha $-feasible attack to a set of sensors $ \mathcal{A}$ is the solution of the following quadratic program with quadratic constraints:
\begin{align}\label{maximization}
		q^m_\alpha&(\mathcal{A},\boldsymbol{A},\boldsymbol{B},\boldsymbol{C})\hspace{-1mm}\triangleq\underset{\boldsymbol{y^a}}\max\,\boldsymbol{y}^{a^T}\boldsymbol{R}_1\boldsymbol{y}^a,\\
		\text{s.t.}\,&  
		\max\left\{\boldsymbol{y}^{a^T}\boldsymbol{P}_1 \boldsymbol{y}^a , \boldsymbol{y}^{a^T}\boldsymbol{P}_2 \boldsymbol{y}^a\right\}\leq \alpha ^2,\label{const3}\\
		&\boldsymbol{y}^a[i] = 0, \quad i \notin \mathcal{A},
\end{align}
where $
		\boldsymbol{R}_1=\boldsymbol{K}^T\boldsymbol{E}\boldsymbol{K},
		\boldsymbol{P}_1=\boldsymbol{S},
		\boldsymbol{P}_2=\boldsymbol{K}^T\boldsymbol{A}^T\boldsymbol{C}^T\boldsymbol{S}\boldsymbol{C}\boldsymbol{A}\boldsymbol{K}.\nonumber$
	\end{theorem}
	\begin{proof}
From Proposition \ref{proposition0}, we know that the maximum CED caused by an impulse attack which we define it as vector $ \boldsymbol{y^a} $ in time step $ k=1 $ to a sensor set $ \mathcal{A}$ is:
\begin{align}\label{maximize}
	q^m(\boldsymbol{y}^a)=\boldsymbol{y}^{a^T}\boldsymbol{K}^T\boldsymbol{E}\boldsymbol{K}\boldsymbol{y}^a,
\end{align}
where $ \boldsymbol{y}^a[i]=0 $ for $ i \notin \mathcal{A}$ and $ y^a[i] $ is the $ i $-th entity of vector $ \boldsymbol{y}^a $. From Lemmas \ref{lemma2} and \ref{lemma3}, we know that the maximum KL distance caused by an $ \alpha $-feasible attack to the sensor set $ \mathcal{A}$ cannot exceed $ \alpha $ and therefore we have:
\begin{align}
	\left(D^{m}(\mathcal{S}^a)\right)^2<\alpha^2,\Rightarrow
	\max\Bigg\{\boldsymbol{y}^{a^T}\boldsymbol{S}\boldsymbol{y}^{a},\left(\boldsymbol{C}\boldsymbol{A}\boldsymbol{K}\boldsymbol{y}^{a}\right)^T\boldsymbol{S}\boldsymbol{C}\boldsymbol{A}\boldsymbol{K}\boldsymbol{y}^{a}\Bigg\}<\alpha^2,\label{constraints}
\end{align}
then, $ \boldsymbol{y}^a $ should maximize \eqref{maximize} with constraints in \eqref{constraints}, and considering $ y^a[i]=0 $ for $ i \notin \mathcal{A}$. Also, since $ \boldsymbol{E} $ and $ \boldsymbol{S} $ are positive-definite matrices then $ \boldsymbol{R}_1 $, $ \boldsymbol{P}_1 $, and $ \boldsymbol{P}_2 $ are all semi-positive definite matrices.
Due to the positive semi-definitiveness, \eqref{maximization} is a convex function and \eqref{const3} is a convex bounded constraint, thus, the solution of \eqref{constraints} will lie at the boundaries\cite{baron1972quadratic}.
	\end{proof}
	Theorem \ref{Theorem1} provides a method for the attacker to find the maximum CED caused by altering a set of sensors without triggering failure alarm. To solve the optimization problem in Theorem \ref{Theorem1}, known techniques such as quadratic programming can be used \cite{baron1972quadratic}. Using Theorem \ref{Theorem1}, we can assign a value to quantify the maximum CED for each of the ICI's SCs. To do so, for each of the ICI's SCs we calculate the following value:
	\begin{align}\label{value}
	\varphi_i(\boldsymbol{A},\boldsymbol{B},\boldsymbol{C})\triangleq q^m_\alpha(\mathcal{N}_i,\boldsymbol{A},\boldsymbol{B},\boldsymbol{C}),
	\end{align}
	where $ \varphi_i $ is the value of SC $ i $ in the state estimation and $ \mathcal{N}_i $ is the set of sensors inside SC $ i $. This value captures the importance of each SC for the attacker and the defender in the first stage of attack, because the attacker can increase the estimation error by $ \varphi_i $ in the second stage of attack after compromising the SC $ i $ in the first stage. Based on this value both the attacker and the defender can prioritize between their actions in the first stage. Since we can now quantify the importance of different SCs under attack, next, we study how the ICI can defend against the first stage of attack during which the attacker and the defender should allocate their available resources on all the SCs based on their values.\vspace{-3mm}
	\section{ICI Security Resource Allocation as a Colonel Blotto Game}\label{section:problem}\vspace{-2mm}
	In this section, we analyze the resource allocation of the attacker and the defender in the first stage of the cyber attack. In our model, the available resources for the defender and attacker are denoted by $ R^d $, and $ R^a $, respectively. Consequently, the defender and the attacker must simultaneously allocate their resources across a finite number of SCs, $ N $. Moreover, each SC $ i $ has a value, $ \varphi_i(\boldsymbol{A},\boldsymbol{B},\boldsymbol{C}) $, given by \eqref{value} \emph{which quantifies the maximum CED} caused by compromising SC $ i $. This value captures both the cyber and physical nature of the ICI as per \eqref{value}. Hereinafter, we use subscripts $a$ and $d$ to denote the attacker and the defender, respectively. 
	
	Also, $ \boldsymbol{r}^j=[r^j_1,\dots,r^j_N]^T $ denotes player $ j $'s allocation vector across $ N $ SCs. In each SC $ i $, the defender assigns a protection level which requires $ r_i^d $ resources. In contrast, the attacker spends some effort to break the sensor's security mechanism, which requires $ r^a_i $ resources in SC $ i $. For instance, in signal watermarking techniques the defender must consider a number of computations in the decoding of each SC's messages in the central server \cite{IoT2018}. To break such watermarking techniques, the attacker must collect the messages of each SC and break the watermarking key using a large number of computation, which requires the attacker to assign a portion of its available computational resources for each SC. Such a resource limitation is not restricted to cases of signal watermarking as it can also be applied to other protection methods such as attack detection filters \cite{pasqualetti2013attack}. 
	
	Therefore, for any protection method, in each SC, if the defender allocates more resources than the attacker then the defender prevents that SC from being compromised. In this case, we assign the normalized value of SC $ i $ to the defender and zero to attacker if the defender wins SC $ i $. In contrast, if the attacker allocates a higher number of resources in each SC, then the attacker can compromise that SC. In this case, we assign the normalized value of SC $ i $ to the attacker and zero to the defender if the attacker wins SC $ i $ (i.e., in this case, the CED is zero, and the defender perfectly protects its SC). Also, in case of equal allocation of resources, which has the probability of zero due to the continuous action space of the attacker and the defender, we share the normalized value of each SC equally between the attacker and the defender. Therefore, in each SC $ i $, the \emph{normalized} payoff for the attacker and defender is given by:\vspace{-0.1cm}
	\begin{align}\label{payoff}
	v^j_i(r^j_i,r^{-j}_i)=
	\begin{cases}
	\phi_i(\boldsymbol{A},\boldsymbol{B},\boldsymbol{C}), & \textrm{if } r^j_i> r^{-j}_i,\\
	\frac{\phi_i(\boldsymbol{A},\boldsymbol{B},\boldsymbol{C})}{2}, & \textrm{if } r^j_i= r^{-j}_i,\\
	0, &\textrm{if } r^j_i< r^{-j}_i,
	\end{cases}
	\end{align}
	where 
	$ -j $ is the opponent of $ j $ and 
	\begin{align}\label{valuedef}
	\phi_i(\boldsymbol{A},\boldsymbol{B},\boldsymbol{C})=\frac{\varphi_i(\boldsymbol{A},\boldsymbol{B},\boldsymbol{C})}{\sum_{m=1}^N\varphi_m(\boldsymbol{A},\boldsymbol{B},\boldsymbol{C})}.
	\end{align}
	The total payoff of the defender and the attacker resulting from allocations across all $ N $ SCs is the sum of the individual payoffs in \eqref{payoff} received from each individual SC:\vspace{-0.15cm}
	\begin{align}\label{defTotalpayoff}
	u^j(\boldsymbol{r}^j,\boldsymbol{r}^{-j})\hspace{-1mm}=\hspace{-1mm}\sum_{i=1}^{N}v^j_i(r^j_i,r^{-j}_i).
	\end{align}
	Here, we define the total maximum CED caused by the allocation vectors $ \boldsymbol{r}^a $ and $ \boldsymbol{r}^d $ as follows:
	\begin{align}\label{gameCED}
	\pi(\boldsymbol{r}^a,\boldsymbol{r}^d)\triangleq u^a(\boldsymbol{r}^a,\boldsymbol{r}^d)\sum_{m=1}^N\varphi_m(\boldsymbol{A},\boldsymbol{B},\boldsymbol{C}),
	\end{align}
	since $ u^a(\boldsymbol{r}^a,\boldsymbol{r}^d) $ captures summation of the estimation errors from all the SCs. The attacker aims to increase its utility function in \eqref{defTotalpayoff} by maximizing {the sum of the compromised SC valuations} which results in maximizing the total state estimation error. Also, the defender seeks to increase its utility function in \eqref{defTotalpayoff} by maximizing {the sum of the valuations of the protected SC} from the cyber attack to minimize the state estimation error. Moreover, the payoff for each player depends on the actions of both players and, thus, we can use a \emph{game-theoretic approach} to solve this problem \cite{bacsar1998dynamic}. In particular, next, we first model the problem as a two-player Colonel Blotto game \cite{Roberson2006} between the attacker and the defender, and then present the solution for the game. The Colonel Blotto game framework is particularly suitable for the considered ICI security problem since, in this game, two colonels simultaneously allocate their available military resources on $ N $ \emph{battlefields}, where the winner of each battlefield is the colonel with a more allocated resources and both the colonels aim to maximize {the sum of the valuations of the} won battlefields. This is similar to the problem in \eqref{defTotalpayoff}, in which SCs are the battlefields and the defender (attacker) maximizes the sum of the valuations of the protected (compromised) SCs.\vspace{-4mm}
	\subsection{Game Formulation and Pure Strategy Nash Equilibrium}\vspace{-2mm}
	To model the interdependent decision making processes of the attacker and defender, we introduce a noncooperative Colonel Blotto game\cite{Roberson2006,GCBG2018} $\splitatcommas{ \Big\{\mathcal{P},\{\mathcal{Q}^j\}_{j\in \mathcal{P}}, \{R^j\}_{j\in \mathcal{P}},N,\{\phi_i^a,\phi_i^d\}_{i=1}^N, \{u^j\}_{j\in\mathcal{P}}\Big\}}$ defined by six components: a) the \emph{players} which are the attacker $ a $ and the defender $ d $ in the set $ \mathcal{P}\triangleq \{a,d\} $, b) the \emph{strategy} spaces $ \mathcal{Q}^j $ for $ j \in \mathcal{P} $, c) \emph{available resource} $ R^j $ for $ j \in \mathcal{P} $, d) \emph{number} of the SCs $ N $, e) normalized \emph{value} of each SC $ i $ for $ j \in \mathcal{P} $, $ \phi^j_i(\boldsymbol{A},\boldsymbol{B},\boldsymbol{C}) $, and f) the \emph{utility function}, $ u^j $, for each player. For both players, the set of \emph{pure} strategies $ \mathcal{Q}^j $ corresponds to the different possible resource allocations across the SCs:
	\begin{align}\label{strategy_set}
	\mathcal{Q}^j=\left\{\boldsymbol{r}^j\Bigg|\sum_{i=1}^{N} r_i^j\leq R^j,r_i^j\geq 0  \right\}.
	\end{align}
	Also, the utility function of each player, $ u^j $, can be defined as in \eqref{defTotalpayoff}. The utility function in \eqref{defTotalpayoff} is a symmetric case for the Colonel Blotto game where $ \phi^d_i(\boldsymbol{A},\boldsymbol{B},\boldsymbol{C})=\phi^a_i(\boldsymbol{A},\boldsymbol{B},\boldsymbol{C})=\phi_i(\boldsymbol{A},\boldsymbol{B},\boldsymbol{C}) $, which indicates that the values of SCs are equal for the defender and the attacker. In the following, first we present the solution of the Colonel Blotto game for a general case of $ \phi^d_i(\boldsymbol{A},\boldsymbol{B},\boldsymbol{C})\neq\phi^a_i(\boldsymbol{A},\boldsymbol{B},\boldsymbol{C}) $, then we derive the solution of symmetric case. For notational simplicity, hereinafter, we drop the arguments $ (\boldsymbol{A},\boldsymbol{B},\boldsymbol{C}) $ in the notation of variables $ \varphi_i $ and $ \phi_i^j $.
	
	One of the most important solution concepts for noncooperative games is that of the \emph{Nash equilibrium} (NE). The NE characterizes a state at which no player $ j $ can improve its utility by changing its own strategy, given the strategy of the other player is fixed. For a noncooperative game, the NE in pure (deterministic) strategies can be defined as follows:\vspace{-2mm}
	\begin{definition}
		A \emph{pure-strategy Nash equilibrium} of a noncooperative game is a vector of strategies $ [{\boldsymbol{r}^a}^*,{\boldsymbol{r}^d}^*] \in \mathcal{Q}^a \times \mathcal{Q}^d $  such that $ \forall j \in \mathcal{P}$, the following holds true:
		\begin{align}
		u^j({\boldsymbol{r}^j}^*,{\boldsymbol{r}^{-j}}^*) \geq 	u^j({\boldsymbol{r}^j},{\boldsymbol{r}^{-j}}^*), \forall {\boldsymbol{r}^j} \in \mathcal{Q}^j.
		\end{align}
	\end{definition}\vspace{-3mm}
	The NE characterize a stable game state at which the defender cannot improve the protection of the ICI's SCs by \emph{unilaterally} changing its action $ \boldsymbol{r}^d $ given that the action of the attacker is fixed at $ {\boldsymbol{r}^a}^* $. At the NE, the attacker cannot increase the state estimation error of the ICI by changing its action, $ \boldsymbol{r}^a $, when the defender keeps its action fixed at $ \boldsymbol{r}^{d*} $. However, the NE is not guaranteed to exist in pure strategies. In particular, for a Colonel Blotto game, without loss of generality, if $ R^d>R^a $, then it can be proven that, for $ NR^a>R^d $ there exist no pure-strategy NE \cite{Roberson2006}. However, it is proven that there exists at least one NE in \emph{mixed strategies} \cite{bacsar1998dynamic} for noncooperative games. When using mixed strategies, each player will assign a probability for playing each one of its pure strategies. For an ICI security problem, the use of mixed strategies is motivated by two facts: a) both players must randomize over their strategies in order to make it nontrivial for the opponent to guess their potential action, and b) the allocation of resources can be repeated over an infinite time duration and mixed strategies can capture the frequency of choosing certain strategies for both players. A mixed strategy, which can be termed as a \emph{distribution of resources}, for player $ j $ is an $ N $-variate distribution function $ G^j : \mathbb{R}_{+}^N \rightarrow [0,1] $ with support contained in player $ j $'s set of feasible allocations, $ \mathcal{Q}^j $. We also define univariate marginal distribution functions (MDFs) $ \{F^j_i\}_{i=1}^N:\mathbb{R}_{+} \rightarrow [0,1] $ for each SC $ i $ and can be called as \emph{distribution of resources on each SC $ i $}. \vspace{-4mm}
	\subsection{Mixed-Strategy Nash Equilibrium Solution}\vspace{-2mm}
	In a game-theoretic setting, each player chooses its own mixed-strategy distribution to maximize its expected utility. We first derive the solution for a special case of our problem in which the attacker and the defender consider the expected allocation of their resources on each SC instead of exact allocation. This is a special case of the Colonel Blotto game known as the \emph{General Lotto} game \cite{kovenock2015generalizations}. In a Colonel Blotto game, the sum of allocated resources cannot exceed the limited resources for the players as in \eqref{strategy_set}. In contrast, in a in General Lotto game, the sum of \emph{expected} allocated resource on SCs cannot exceed the restricted resource of players:\vspace{-1mm}
	\begin{align}\label{expres}
	\mathcal{Q}^j=\left\{\boldsymbol{r}^j\Bigg|\sum_{i=1}^{N} \mathbb{E}_i^j(r)\leq R^j,r_i^j\geq 0  \right\}.
	\end{align} 
	where $  \mathbb{E}_i^j(r) $ is the expected value of resources allocated by player $ j $ on SC $ i $. In this case, the utility of each player $ j \in \mathcal{P} $ is defined as the expected value over its mixed strategies:\vspace{-1mm}
	\begin{align}\label{expected_payoff}
	U^j(G^j,G^{-j})=	U^j(\{F^j_i\}_{i=1}^N,\{F^{-j}_i\}_{i=1}^N)
	=\sum_{i=1}^N\left[\int_{0}^{\infty}\phi_i^jF^{-j}_i(r^j_i)dF^j_i\right]. 
	\end{align}
	Thus, player $ j $'s optimization problem considering its constraint on the available resource is:
	\begin{align}\label{budget_constraint}
	\max_{\{F^j_i\}_{i=1}^N}\sum_{i=1}^{N}\left[\int_{0}^{\infty}\left[\phi_i^jF^{-j}_i(r^j_i)-\zeta^jr^j_i\right]dF^j_i\right]+\zeta^jr^j,
	\end{align}
	where $ \zeta^j $ is a multiplier for player $ j $'s expected resource allocation constraint. For each $ i=1,\dots,N $, the corresponding first-order condition for maximizing \eqref{budget_constraint} is given by:
	\begin{align}
	\frac{d}{dr^j_i}\left[\phi_i^jF^{-j}_i(r^j_i)-\zeta^j\right]=0,\Rightarrow
	\label{allpay}
	\frac{\phi_i^j}{\zeta^j}\frac{d}{dr^j_i}F^{-j}_i(r^j_i)=1,
	\end{align}
	where \eqref{allpay} is equivalent to the necessary condition for a single all-pay auction game where player $ j $'s value for the prize in auction is $ \frac{\phi_i^j}{\zeta^j} $\cite{baye1996all}. In such an all-pay auction, if $ \frac{\phi_i^j}{\zeta^j} \geq \frac{\phi_i^{-j}}{\zeta^{-j}} $ the solution of \eqref{allpay} is described as follows:
	\begin{align}\label{Fdist1}
	F^{-j}_i(r)=\left(\frac{\frac{\phi_i^j}{\zeta^j} - \frac{\phi_i^{-j}}{\zeta^{-j}}}{\frac{\phi_i^j}{\zeta^j}}\right)+\frac{r}{\frac{\phi_i^j}{\zeta^j}}, \, \, \, \,
	F^{j}_i(r)=\frac{r}{\frac{\phi_i^{-j}}{\zeta^{-j}}},\hfil r\in\left[0, \frac{\phi_i^{-j}}{\zeta^{-j}}\right].
	\end{align}
	Now, to find the multipliers $ (\zeta^a,\zeta^d) $, let $ \mu \triangleq \frac{\zeta^a}{\zeta^d} $ and assume that $ \Omega_a(\mu) $ is the set of SCs in which $ \frac{\phi_i^a}{\phi_i^d}>\mu $. Then using \eqref{budget_constraint}, \eqref{Fdist1}, we have:
	\begin{align}\label{constd}
	\sum_{i \in \Omega_a(\mu)}\frac{\phi_i^d}{2\zeta^d}+\sum_{i \notin \Omega_a(\mu)}\frac{\left(\frac{\phi_i^a}{\zeta^a}\right)^2}{2\left(\frac{\phi_i^d}{\zeta^d}\right)}&=R^a,\\\label{consta}
	\sum_{i \in \Omega_a(\mu)}\frac{\left(\frac{\phi_i^d}{\zeta^d}\right)^2}{2\left(\frac{\phi_i^a}{\zeta^a}\right)}+\sum_{i \notin \Omega_a(\mu)}\frac{\phi_i^a}{2\zeta^a}&=R^d.
	\end{align}
	From \cite[Propostion 1]{kovenock2015generalizations} we know that there exists at least one solution to \eqref{constd} and \eqref{consta}. 
	
	Now that we characterized the functions that maximize the expected utility of players in \eqref{expected_payoff}, we first define the solution concept of \emph{mixed strategy Nash equilibrium} (MSNE) and then, finalize the solution of Lotto game by deriving its MSNE. The MSNE is defined as follows:\vspace{-2mm}
	\begin{definition}\vspace{-2mm}
		A mixed strategy profile $ G^* $ constitutes a mixed strategy Nash equilibrium if for player $ j $ we have:
		\begin{align}
		U^j({G^j}^*,{G^{-j}}^*)\geq U^j(G^j,{G^{-j}}^*) \,\, \forall G^j \in \mathcal{G}^j
		\end{align} 
		where $ \mathcal{G}^j $ is the set of all probability distributions for player $ j $ over its action space $ \mathcal{Q}^j $.
	\end{definition}
	The MSNE for this game characterizes a state of the system at which the defender has chosen its optimal randomization over the allocation of resources on SCs and, thus, cannot improve the protection of ICI's SCs by changing this choice. Also, the MSNE for the attacker is a probability distribution that captures the allocation of its resources over the SCs in a way to maximize the state estimation error when the defender chooses its MSNE strategies. Using the definition of the MSNE, we define the expected CED at MSNE as follows:
	\begin{align}
	\Pi({G^a}^*,{G^d}^*)=\Pi(\{{F^a}^*_i\}_{i=1}^N,\{{F^d}^*_i\}_{i=1}^N)\label{expCEDgame}\triangleq U^a({G^a}^*,{G^d}^*)\sum_{m=1}^{N}\varphi_i.
	\end{align}
	It is proven in \cite[Theorem 1]{kovenock2015generalizations} that for each solution $ (\zeta^d,\zeta^a) $ for system of equations in \eqref{constd} and \eqref{consta}, each player in a General Lotto game has a unique MSNE with univariate marginal distributions in \eqref{Fdist1}. In the following {remark}, we characterize the solution for our problem when the values of the ICI's SCs for both attacker and defender are equal and, then, we find the expected state estimation error.
\begin{remark}\label{proposition1}
		For the problem of resource allocation over SCs having equal values for the attacker and defender $\phi_i^a=\phi_i^d\triangleq\frac{\varphi_i(v)}{\sum_{i=1}^N\varphi_i(v)}$, at the MSNE, the MDFs for the attacker and defender, when the defender's resources are greater than the attacker's resources, $ R^d\geq R^a $, will be given by:
		\begin{align}\label{sol1}
		{F_i^a}^*(r)=\left(1-\frac{R^a}{R^d}\right)+\frac{r}{2\phi_iR^d}\frac{R^a}{R^d}, \,\,\, {F_i^d}^*(r)=\frac{r}{2\phi_iR^d},\hfill r\in [0,2\phi_iR^d],
		\end{align} 
		and the expected CED at MSNE will be $ \frac{R^a}{2R^d}\sum_{i=1}^{N}\varphi_i $.
		Considering $ \phi_i^a=\phi_i^d$  and $ R^d\geq R^a $ the solution of \eqref{consta} and \eqref{constd} is $ \zeta^a=\frac{1}{2R^d}, \zeta^d=\frac{R^a}{2(R^d)^2} $. Then, by substituting  $ \zeta^a $ and $ \zeta^d $ into \eqref{Fdist1} we can directly prove \eqref{sol1} and find the expected CED at the MSNE.
	\end{remark} \vspace{-3mm}
	The value $ \frac{R^a}{2R^d}\sum_{i=1}^{N}\varphi_i $ captures the ICI's expected estimation error for the case in which, the available resource for the attacker and defender are $R^a $ and $ R^d $, respectively. From Proposition \ref{proposition1}, we can conclude two important points: a) the probability of allocation of resources greater than $ 2\phi_iR^d $ to SC $ i $ is zero, $ \textrm{Pr}\left(r_i^j>2\phi_iR^d\right)=0 $, b) as the ratio of the attacker's available resource to the defender's available resource, $ \frac{R^a}{R^d} $, increases, the expected state estimation error increases. Next, we prove that, if the defender concentrates only on one of the CIs without considering their interdependence, then the expected estimation error increases.\vspace{-3mm}
	\begin{theorem}\label{Theoreminterdependence}\vspace{-2mm}
		Suppose the sets $ \mathcal{N}^e, \mathcal{N}^g,\mathcal{N}^w $ contain the SCs inside the power, natural gas, and water CIs. Then, the ICI's expected estimation error increases if the defender does not consider the interdependence between the CIs.  \vspace{-3mm}
	\end{theorem}
	\begin{proof}
		Without loss of generality, suppose that the defender only protects the water CI while the attacker considers all of the ICI's SCs. In this case, we assume that the value of the SCs in the natural gas and power CI for the defender is $ \varepsilon \rightarrow 0 $. The water CI's SCs, however, will have new values as $	\phi^d_i=\frac{\varphi_i}{\sum_{m \in \mathcal{N}^w}\varphi_m}, \, \forall i \in \mathcal{N}^w,$ while the values of SCs for the attacker are similar to \eqref{valuedef}. Then, the ratio of the values for the attacker and defender is as follows:
		\begin{align}
		\frac{\phi^a_i}{\phi^d_i}=\begin{cases}
		\kappa\triangleq\frac{\sum_{m\in\mathcal{N}^w}\varphi_m}{\sum_{m=1}^{N}\varphi_m}<1, & \textrm{if } i\in \mathcal{N}^w,\\
		\frac{\varphi_i}{\varepsilon}\rightarrow \infty, & \textrm{if } i \notin \mathcal{N}^w.
		\end{cases}
		\end{align}
		To find the values of $ \zeta^a $ and $ \zeta^d $ we consider two cases for $ \mu $:\\
		If $ \mu<\kappa $, then $ \Omega_a(\mu)=\mathcal{N}$, and, thus, from \eqref{constd} and \eqref{consta} we have $ \zeta^d=\frac{1}{2R^a} $, $ \zeta^a=\frac{\kappa R^d}{2(R^a)^2} $ which results in $ \mu = \frac{\kappa R^d}{R^a} $ where it violates the condition $ \mu<\kappa $. \\
		If $ \mu\geq\kappa $, then $ \Omega_a(\mu)=\mathcal{N}^g\bigcup \mathcal{N}^e$, and from \eqref{constd} and \eqref{consta} we find $ \zeta^a=\frac{\kappa}{2R^d} $, $ \zeta^d=\frac{R^a}{2(R^d)^2} $ which results in $ \mu=\frac{\kappa R^d}{R^a}\geq \kappa $. Therefore, the MDFs for the attacker and the defender are:
		\begin{align}
		\bar{F}_i^a(r)&=
		\begin{cases}
		1, &\hspace{-4mm}\begin{aligned}
		r=0,i\notin \mathcal{N}^w
		\end{aligned}\\
		\hspace{-1mm}\left(1-\frac{R^a}{R^d}\right)+\frac{rR^a\kappa}{2(R^d)^2\phi_i}, &\hspace{-3.5mm}
		r\hspace{-1mm}\in\hspace{-1mm} \left[0,\frac{2\phi_iR^d}{\kappa}\right]\hspace{-1mm},\hspace{-0.5mm}
		i\in \mathcal{N}^w
		\end{cases}\\
		\bar{F}_i^d(r)&=
		\begin{cases}
		1,& r=0, i\notin \mathcal{N}^w\\
		\frac{r\kappa}{2R^d\phi_i}, & r\in \left[0,\frac{2\phi_iR^d}{\kappa}\right],  i\in \mathcal{N}^w,
		\end{cases}
		\end{align}	
		Using above distribution functions and \eqref{expCEDgame}, we can find the expected CED at the MSNE as $
		\Pi(\{\bar{F}^d_i\}_{i=1}^N,\{\bar{F}^{a}_i\}_{i=1}^N)=\left(\frac{R^a\kappa}{2R^d}+\frac{1-\kappa}{2}\right)\sum_{i=1}^{N}\varphi_i.$ To reduce the expected CED in Proposition \ref{proposition1}, we must have:
		\begin{align}
		\left(\frac{R^a\kappa}{2R^d}+\frac{1-\kappa}{2}\right)\sum_{i=1}^{N}\varphi_i<\frac{R^a}{2R^d}\sum_{i=1}^{N}\varphi_i, \Rightarrow
		\left(\kappa-1\right)\left(\frac{R^a}{R^d}-1\right)<0,
		\end{align}
		and since $ \kappa<1 $ then we need to have $ \frac{R^a}{R^d}>1 $ which results in a contradiction because we know that $ \frac{R^a}{R^d}\leq1 $. Therefore, the defender is never better off if it does not allocate resources to the water infrastructure. Hence, when the defender allocates its resources only on a subset of CIs, although the expected allocated resources on these CIs increases, the estimation error of their states will also increase due to the interdependence of the state variables of the ICI.
	\end{proof}
	Theorem \ref{Theoreminterdependence} illustrates the role of the interdependence between power, natural gas, and water CIs in the state estimation. The defender must consider all the CIs and their interdependence in the security analysis. Otherwise, if the defender only protects one of the CIs, then the attacker can cause higher estimation error on all the state variables of the ICI. 
	
The solutions presented in Proposition \ref{proposition1} and Theorem \ref{Theoreminterdependence} are for the General Lotto game in which the constraints on resources hold true in expectation as in \eqref{expres}. Next, we analyze a special case for the values of the SCs in which the solution of the General Lotto game in Proposition \ref{proposition1} can be applied to our original CBG. {We define $ n_\delta $ as the total number of distinct SCs with distinct valuations and show the following result:\vspace{-5mm}
\begin{theorem}\label{TheoremBlo}
For $ n_\delta \rightarrow \infty$, the derived solution of General Lotto game in Remark \ref{proposition1} can be applied to the original CBG for ICI.\vspace{-5mm}
\end{theorem}
\begin{proof}
From \cite[Proposition 2]{kovenock2015generalizations}, given a solution $ (\zeta^a,\zeta^d) $ to \eqref{constd} and \eqref{consta}, if for each pair of valuations $ (\phi^a_i,\phi_i^d) $ we have $ \frac{\phi^a_i \zeta^d}{\phi^d_i\zeta^a}\leq 1 $, and $ \frac{2}{n_\delta} \leq \frac{\phi^a_i \zeta^d}{\phi^d_i\zeta^a} $, then there exists a Nash equilibrium of the CBG with the same set of univariate marginal distributions and expected payoffs in the General Lotto game. In our defined game $ \phi^d_i = \phi^a_i\,\, \forall i = \mathcal{N} $, we have derived in Remark \ref{proposition1} that $ \zeta^a=\frac{1}{2R^d}, \zeta^d=\frac{R^a}{2(R^d)^2} $. Hence, we will have $ \frac{\phi^a_i \zeta^d}{\phi^d_i\zeta^a} = \frac{ \frac{R^a}{2(R^d)^2}}{\frac{1}{2R^d}}= \frac{R^a}{R^d} \leq 1 $. Therefore, for a large-scale ICI where $ n_\delta \rightarrow \infty $, then we will have $ \frac{2}{n_\delta}\rightarrow 0  \leq \frac{R^a}{R^d} = \frac{\phi^a_i \zeta^d}{\phi^d_i\zeta^a} $, which proves Theorem \ref{TheoremBlo}.
\end{proof}}

In this section, we derived two key theorems. Theorem \ref{Theoreminterdependence} highlights the importance of interdependence in the security analysis and proves that the defender cannot minimize the state estimation error by allocating resources to a single CI. Moreover, we characterized the solution of our Colonel Blotto game for the ICI model in Theorem \ref{TheoremBlo} and we derived the expected estimation error at the MSNE as a function of the attacker's and the defender's available resources. 	\vspace{-3mm}
	\section{Simulation Results and Analysis} \label{sect:simul} \vspace{-2mm}
	{For our simulations, we study the ICI example in Fig. \ref{fig_ICI} which captures a real-world ICI scenario such as in \cite{nan2017quantitative,alamian2012state,andersson2012dynamics}, and \cite{Burgschweiger2009}.
	In Fig. \ref{fig_ICI}, we consider 10 generators out of which $6$ are supplied by natural gas and $ 4 $ require water flow to operate. Here, $ 11 $ natural gas pipelines, and $ 11 $ water pipelines distribute natural gas and water to the demand junctions. Based on this example we find the matrices $ \boldsymbol{A} $ and $ \boldsymbol{B} $ in \eqref{eq:interCI} using Appendix \ref{ICIAppendix} and simulate the ICI. To illustrate how the changes in one CI can affect the state variables of other CIs, we increase the power demand in generator $ 5 $, $ u_5^e $, at time $ t=20.5 $. Fig. \ref{fig_interdependence} shows the change of state variables of the natural gas pipeline between junctions $ 2 $ and $ 3 $ and state variables of water pipeline between the junctions $ 2 $ and $ 3 $. The reason is that any increase in power demand results in an increase of electric power generation, and due to the interdependence between electric power generation and the consumption of the natural gas and water, the state variables of the natural gas and water CIs change. }In this ICI, each power generator and each junction in  the natural gas and water systems has a demand profile, which specifies power, gas, and, water demand, at each time. We design $ \boldsymbol{A} $ and $ \boldsymbol{B} $ such the dynamic ICI model constitutes an asymptotically stable system. Also, we consider 32 SCs which collect sensor data from different physical components. Based on the sensor network architecture in Fig. \ref{fig_ICI}, we generate the matrix $ \boldsymbol{C} $. In our simulations, we consider $ 0.5 $-feasible attacks that we use to compute the values of the SCs, $ \phi_i(\boldsymbol{A},\boldsymbol{B},\boldsymbol{C}) $. {Moreover, in our example all of the SCs are distinct, thus, for $ R^a/R^d > 1/32 $ the solution of Blotto game in Theorem \ref{TheoremBlo} will hold.}
\begin{figure*}
	\begin{subfigure}{0.66\textwidth}
		\captionsetup{singlelinecheck = false, justification=justified}
		\centering
		\includegraphics[width=\columnwidth]{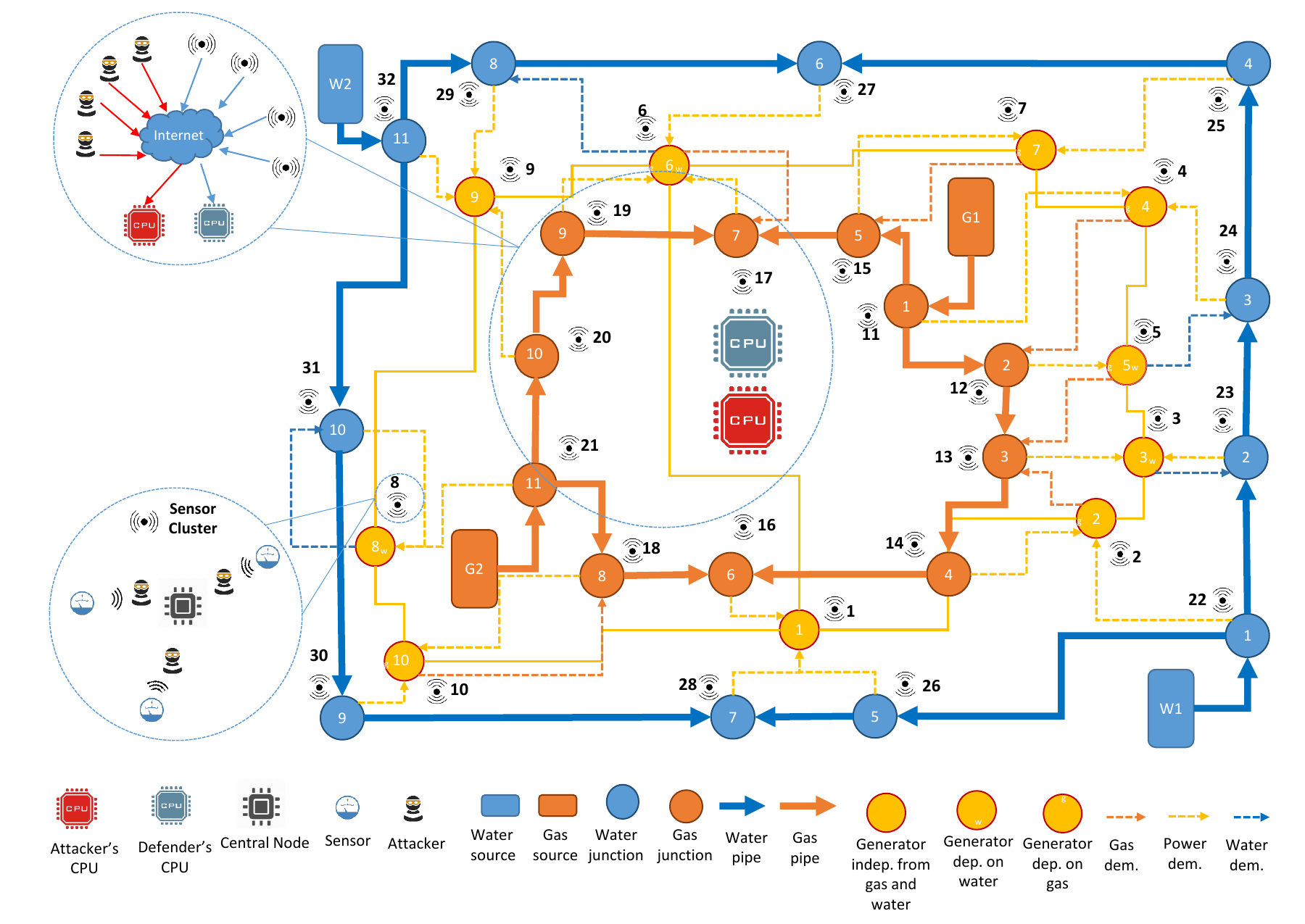}
		\vspace{-9mm}
		\caption{An illustrative example of an ICI.}
		\label{fig_ICI}
	\end{subfigure}
	~
	\begin{subfigure}{0.32\textwidth}
		\centering
		\captionsetup{singlelinecheck = false, justification=justified}
		\includegraphics[width=\columnwidth]{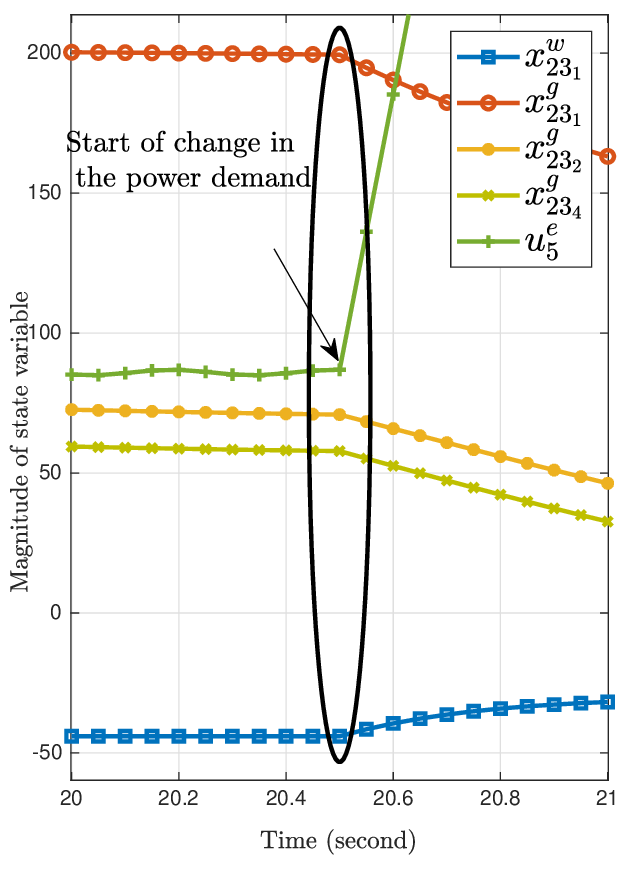}
		\caption{Interdependence between the CIs.}
		\label{fig_interdependence}
	\end{subfigure}
\vspace{-5mm}
\caption{The interdependence of power, natural gas, and water CIs and their state variables.}
\vspace{-5mm}
\end{figure*}
\begin{figure*}	
	\centering
	\begin{subfigure}{0.38\textwidth}
		\centering
		\captionsetup{singlelinecheck = false, justification=justified}
		\includegraphics[width=\columnwidth]{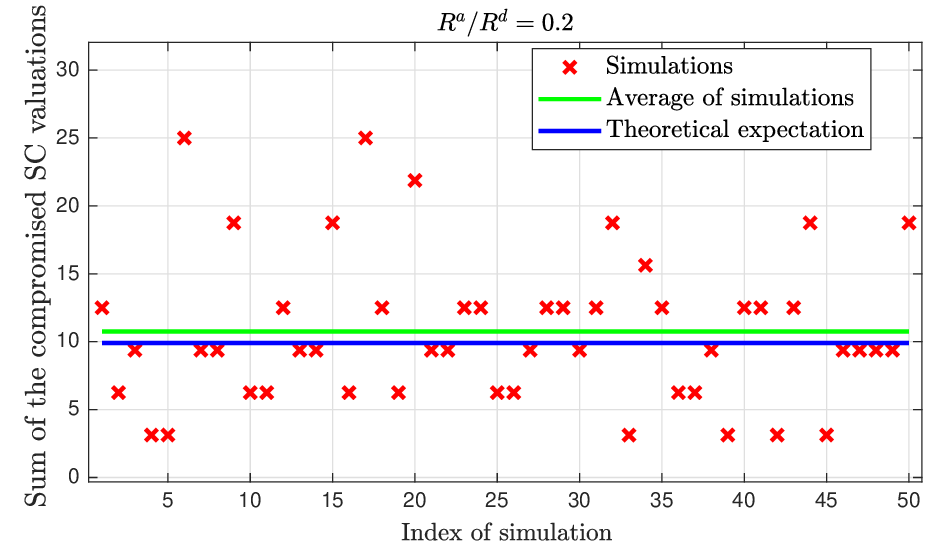}\vspace{-2mm}
		\caption{Sum of the compromised SC valuations.}
		\label{fig_percentageofCompromisedSCs}
		\vspace{-3mm}
	\end{subfigure}~
	\begin{subfigure}{0.38\textwidth}
		\centering
		\captionsetup{singlelinecheck = false, justification=justified}
		\includegraphics[width=\columnwidth]{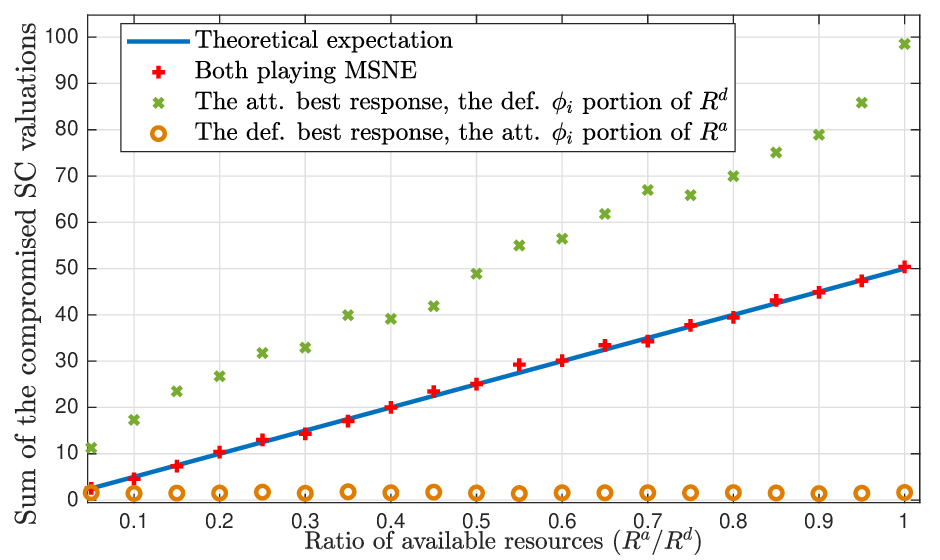}\vspace{-2mm}
		\caption{Average sum of the compromised SC valuations.}
		\label{fig_RaRange}
		\vspace{-3mm}
	\end{subfigure}
	\caption{The outcome of the CBG with the strategies at the MSNE.}
	\vspace{-9mm}
\end{figure*}

In Fig. \ref{fig_percentageofCompromisedSCs}, we simulate the presented ICI in Fig. \ref{fig_ICI} for the case in which the ratio of the available resources is $ \frac{R^a}{R^d}=0.2 $. Over 50 simulation runs, both players, empirically, play their MSNE based on the MDFs in Proposition \ref{proposition1}. The average percentage of compromised SCs in this case is $ 11\%$ which is very close to the theoretical expected percentage of the compromised SCs which is $ \frac{R^a}{2R^d}=10\% $. Fig. \ref{fig_percentageofCompromisedSCs} shows that since the attacker and the defender randomize between their strategies in each simulation, the percentage of compromised SCs of each simulation might differ from the theoretical expected percentage of compromised SCs. However, the average percentage of compromised SCs for 50 simulations conforms to the theoretical prediction.

	\begin{figure*}
		\centering
		\begin{subfigure}{0.38\textwidth}
			\centering
			\captionsetup{singlelinecheck = false, justification=justified}
			\includegraphics[width=\textwidth]{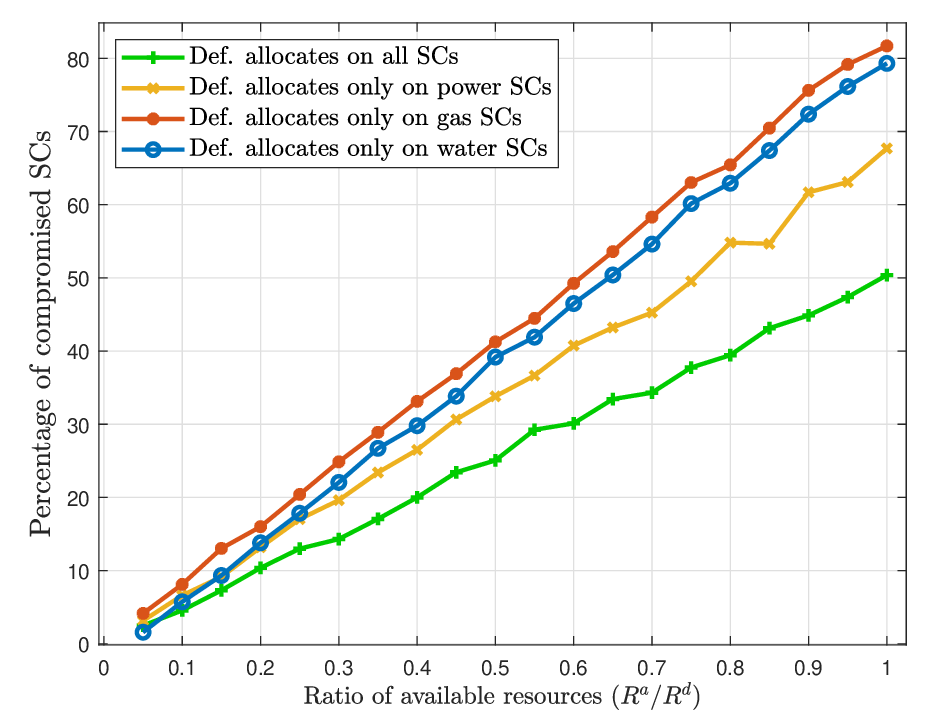}
			\vspace{-6mm}
			\caption{The average percentage of compromised SCs in ICI when the defender allocates resources to only one SC.}
			\label{fig:sub1}
			\vspace{-6mm}
		\end{subfigure}
		~
		\begin{subfigure}{0.38\textwidth}
			\centering
			\captionsetup{singlelinecheck = false, justification=justified}
			\includegraphics[width=\textwidth]{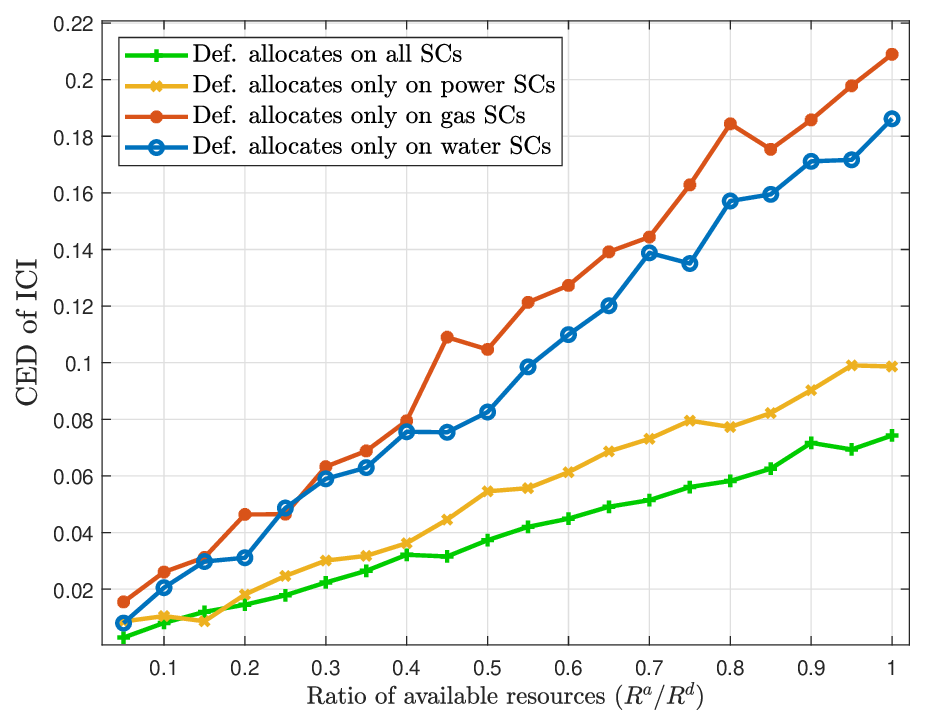}
			\vspace{-6mm}
			\caption{The average CED caused by the cyber attack to ICI  when the defender allocates resources to only one SC.}
			\vspace{-6mm}
			\label{fig:sub2}
		\end{subfigure}
		\caption{The effect of interdependency on the security solution.}\vspace{-12mm}
	\end{figure*}

	To analyze the efficiency of the MSNE , we consider three baseline approaches: a) both players play their MSNE (which is the solution of the game), b) the attacker allocates a portion $ \phi_i $ of its resource on SC $ i $ while the defender plays \emph{best response}, and c) the defender allocates a portion $ \phi_i $ of its resource on SC $ i $ while the attacker plays a best response. The \emph{best response} of player $ j $ is the pure strategy which maximizes its utility function. Fig. \ref{fig_RaRange} shows that, when the defender allocates a portion $ \phi_i $ of its resource on SC $ i $ without mixing the allocation of resources, the attacker can predict the defender's strategy and {can get higher payoff}. In contrast, if the attacker allocates a portion $ \phi_i $ of its resources on SC $ i $ without playing a mixed strategy, then the defender can protect almost all the SCs since it knows the exact amount of attacker's allocated resource on each SC. Thus, the attacker is clearly better off randomizing using an MSNE. Therefore, Fig. \ref{fig_RaRange} shows that if the defender plays MSNE, it can protect the SCs at least $ 50 \% $ better than the case in which it allocates a portion $ \phi_i $ of its resource to every SC $ i $. {The $ 50\% $ improvement actually can be derived by comparing the slopes of lines passing through the simulation points in Fig. \ref{fig_RaRange}.} Moreover, Fig. \ref{fig_RaRange} shows that, as much the portion of the attacker's resources to the defender's resources increases, {the average sum of the compromised SC valuations} increases.  
	
	To analyze the interdependence between the three CIs, we simulate our model for cases in which the defender protects only one of the CIs. Fig. \ref{fig:sub1} shows the average percentage of compromised SCs. From Fig. \ref{fig:sub1}, we can see that the defender loses more SCs to the attacker by concentrating on the security of only one CI. Moreover, Fig. \ref{fig:sub2} shows the average CED caused by the cyber attack on the ICI. Fig. \ref{fig:sub2} indicates that the state estimation of each CI depends on the data of the SCs in other CIs. Therefore, when the defender focuses only on one CI to protect, the attacker can disturb this CI's state estimation by attacking the SCs in other CIs. The simulation results in Figs. \ref{fig:sub1} and \ref{fig:sub2} corroborate the theoretical results in Theorem \ref{Theoreminterdependence}, where we proved that the defender cannot protect the estimation error by allocating its resources only on one CI without considering the interdependence between the CIs.
	
	\begin{figure*}
	\centering
	\begin{subfigure}{0.44\textwidth}
		\centering
		\captionsetup{singlelinecheck = false, justification=justified}
		\includegraphics[width=0.9\columnwidth]{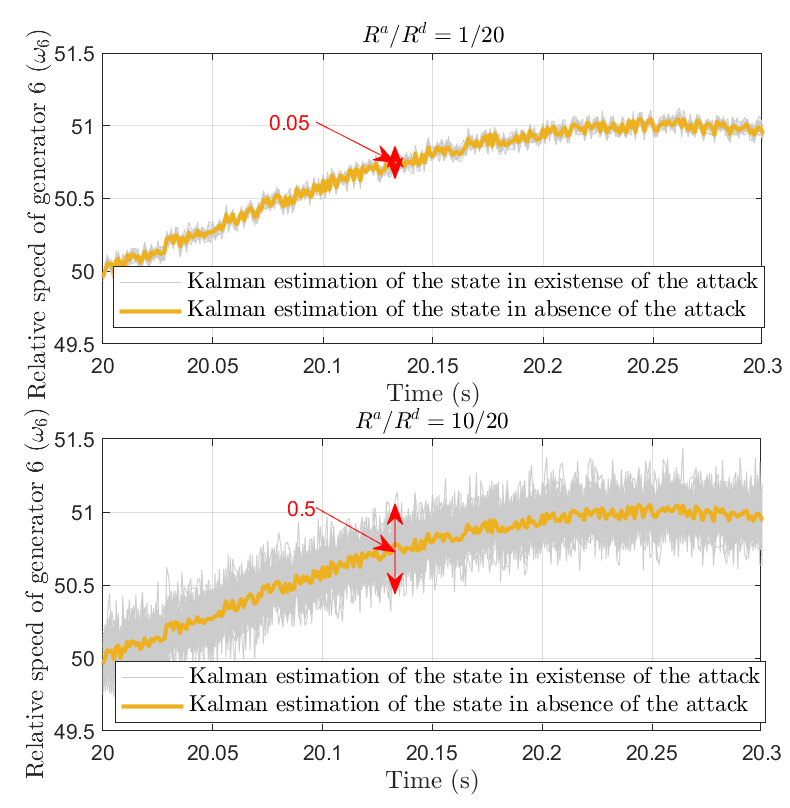}\vspace{-2mm}
		\caption{The average estimation error caused by the attack to the ICI at MSNE for two cases: $ \frac{R^a}{R^d}=\frac{1}{20},\frac{10}{20} $.}
		\vspace{-3mm}
		\label{fig_Kalman}
	\end{subfigure}~
	\begin{subfigure}{0.44\textwidth}
		\centering
		\captionsetup{singlelinecheck = false, justification=justified}
		\includegraphics[width=0.9\columnwidth]{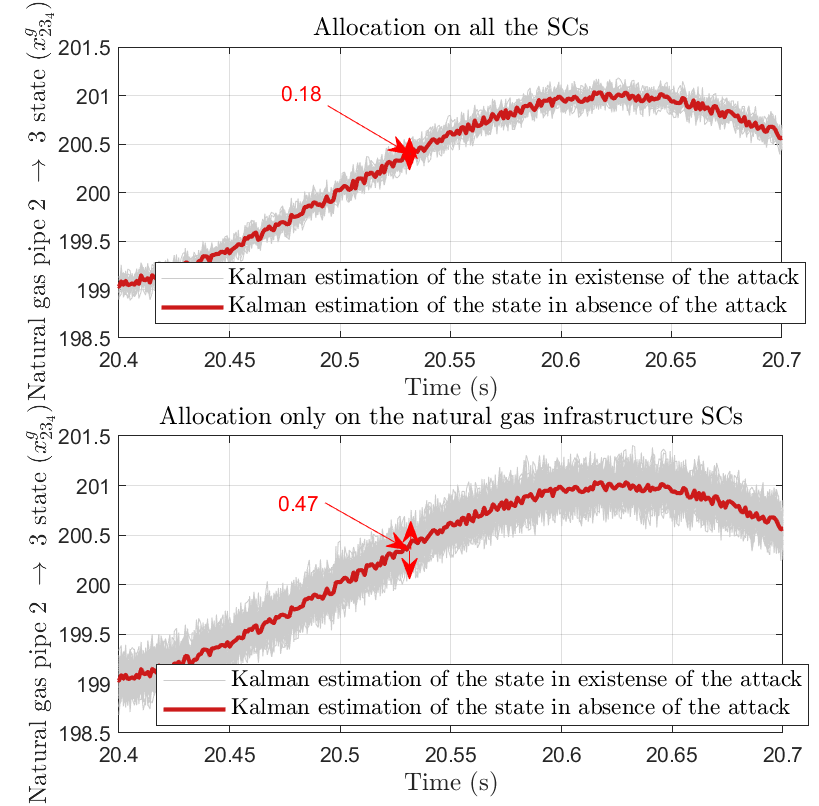}\vspace{-2mm}
		\caption{The average estimation error of a gas state caused by the attack on the ICI when the defender protects the entire ICI VS. only the natural gas CI.}
		\label{fig_kalmaninter}
		\vspace{-3mm}
	\end{subfigure}
	\caption{The average estimation error when the ICI is attacked.}\vspace{-12mm}
\end{figure*}
	
	In Fig. \ref{fig_Kalman}, we simulate the use of a KF for the estimation of one of the ICI's state variables, $ \omega_6 $, in presence of the attack and in absence of the attack, for two values of $ \frac{R^a}{R^d} $. Fig. \ref{fig_Kalman} shows that the estimation of the state variable $ \omega_6 $ causes 10 times higher error in the case of $ \left(\frac{R^a}{R^d}\right)_1=\frac{10}{20} $ compared to the case of $ \left(\frac{R^a}{R^d}\right)_2=\frac{1}{20} $ when both players play their MSNE. To explain this result, from Proposition \ref{proposition1}, we know that the expected estimation error of the attack is $ \Pi^a=\frac{R^a}{2R^d}\sum_{i=1}^{N}\varphi_i $. Hence, the error caused by the attacker increases by a factor of $ \frac{\left(\frac{R^a}{R^d}\right)_1}{\left(\frac{R^a}{R^d}\right)_2}=10 $.
	
	In Fig. \ref{fig_kalmaninter}, we analyze the estimation error of one of the state variables of the pipeline between junction 2 and 3 of natural gas CI. In this simulation, we examine two games where in first case, we consider that the defender can allocate resources to all the SCs of ICI, while in the second case, we consider that the defender can only protect the SCs of the natural gas CI. Note that the ratio between the players' resources is $ \frac{R^a}{R^d}=0.25 $ in two cases. Fig. \ref{fig_kalmaninter} shows that, at MSNE although the defender protects only the SCs of the natural gas CI, the average estimation error of a natural gas CI state variable in this case is larger than the estimation error in the case of protecting all the SCs of ICI. Moreover, although the allocated resources on natural gas CI is increased but from Fig. \ref{fig_kalmaninter} we see that the average estimation error in natural gas state variable is increased with the factor of $ \frac{0.47}{0.18}=2.6 $. This difference between the error estimation of two cases is very close to the theoretical expected utility of the attacker in Theorem \ref{Theoreminterdependence}. From Theorem \ref{Theoreminterdependence}, we know that the ratio of the expected utility of the attacker in two cases is $ \frac{\left(\frac{R^a\kappa}{R^d}+1-\kappa\right)}{\frac{R^a}{R^d}} $. Here, $ \kappa=\frac{\sum_{i\in \mathcal{N}^g}\varphi_i}{\sum_{i=1}^{N}\varphi_i}=0.38 $. Therefore, the ratio of the expected utility of the attacker in two cases is $ 2.86 $. Fig. \ref{fig_kalmaninter} illustrates the interdependence of three CIs. To protect the states of only one of the CIs, the defender has to consider the interdependence of CIs, and allocate its resources on all the SCs of ICI not only the CI which it desires to protect. \vspace{-5mm}
	\section{conclusion}\label{sect:conc}\vspace{-2mm}
	In this paper, we have analyzed the problem of allocating limited protection resources on sensor areas (SCs) of an ICI using a game-theoretic approach. We have modeled the dynamic system of interdependent power, natural gas, and water infrastructure. We have analyzed the state estimation of the states of an ICI and the maximum reachability of estimation errors of different SCs of the ICI. We have considered a general model of protection of SCs of an ICI when the available resources of the defender and the attacker are limited. In particular, we have formulated the problem of allocating the limited resources of the defender and the attacker as a Colonel Blotto game. We have then derived the MSNE of the defender and the attacker in closed-form as a function of the values of SCs and the available resources for the attacker and the defender. The derived MSNE gives insights on the allocation of the resources on each SC and also underlines the interdependence of the three infrastructure. Simulation results verify that the derived MSNE is the defender's best strategy and due to the interdependence of three CIs, the defender must consider ICI as a unified system in the security analysis.\vspace{-2mm}
\begin{appendices}{\vspace{-4mm}
\section{Critical Infrastructure Dynamic Model}\label{ICIAppendix}\vspace{-3mm}
In this appendix, we derive a dynamic system model for the power, natural gas, and water CI. Moreover, we analyze the dynamic model of interdependence between these three CIs.\vspace{-4mm}
\subsection{Power Infrastructure}\vspace{-3mm}
To analyze the power system, we consider a synchronous generator connected to a transmission line as the study system. For a large-scale power system, consisting of $ n_e $ generators interconnected through a transmission network, the model derived in \cite{andersson2012dynamics} is used. In this model, each generator is considered as a subsystem, with the input to each subsystem $ i $ being the power demand from the connected bus, $ P_{e_i} $. Any changes in the power demand, $ P_{e_i}, $ will result in a change in the frequency of the generator and the mechanical input to the generator. The block diagram of two connected generators is shown in Fig. \ref{intergenerator}. In this model, the dynamics of each subsystem $ i $ can be written as follows:
\small
\begin{align}\label{eq:omega}
\dot{\omega}_i(t)&=-\frac{D_i}{J_i}\omega_i(t)+\frac{1}{J_i}\hspace{-0.1cm}\left(\hspace{-0.1cm}P_{m_i}(t)-P_{e_i}(t)-\hspace{-0.3cm}\sum_{j=1, j\neq i}^{n_e}P_{ij}\right)\hspace{-0.1cm}.
\end{align}
\begin{align}\label{eq:mechpower}\small
\dot{P}_{m_i}(t)=-\frac{1}{T_{t_i}}P_{m_i}(t)-\frac{1}{P^o_iT_{t_i}}\omega_i,
\end{align}
\begin{align}\label{eq:line}\small
\dot{P}_{ij}=P^o_{ij}(\omega_i-\omega_j),
\end{align}\normalsize
where \eqref{eq:omega}, \eqref{eq:mechpower}, and \eqref{eq:line} represent the mechanical, feedback, and line dynamics of each subsystem \cite{andersson2012dynamics} and $ P^o_{ij }$ is the linearized power flow in the constant voltage that can be derived as $
P^o_{ij}=\frac{E_iE_j}{x_{l_{ij}}}\cos(\delta^o_i-\delta^o_j).
$
\begin{figure}[!t] 
	\centering
	\begin{subfigure}{0.46\columnwidth}
    	\centering
    	\includegraphics[width=0.7\columnwidth]{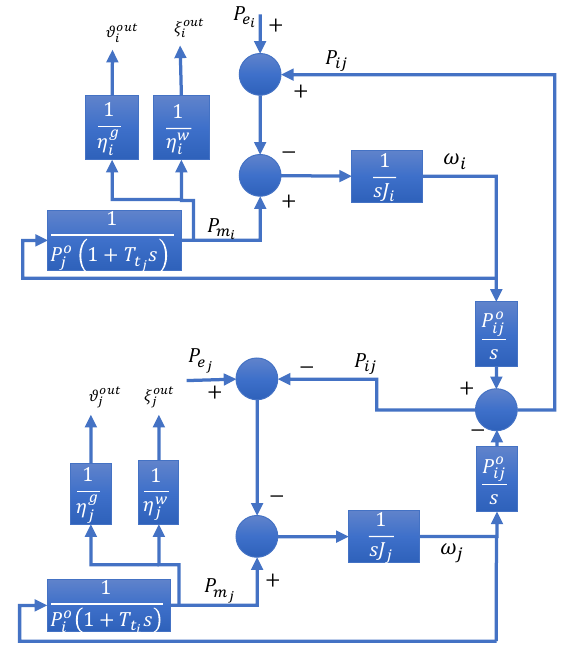}
    	\caption{Two connected generators.}
    	\label{intergenerator}
    \end{subfigure}
    \begin{subfigure}{0.5\columnwidth}
    \centering
    \begin{subfigure}{\columnwidth}
    	\centering
    	\includegraphics[width=0.7\columnwidth]{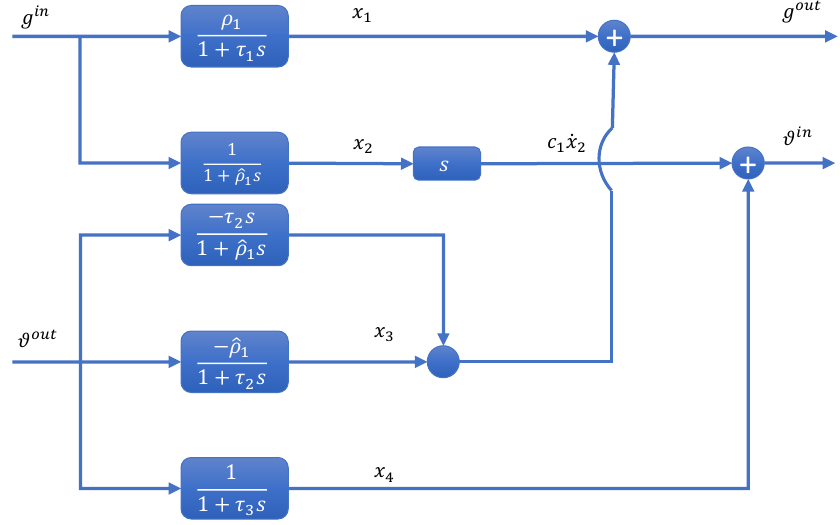}
    	\caption{Natural gas pipeline.}
    	\label{fig_gas}
    \end{subfigure}\hspace{5mm}
    \begin{subfigure}{\textwidth}
    	\centering
    	\includegraphics[width=0.7\columnwidth]{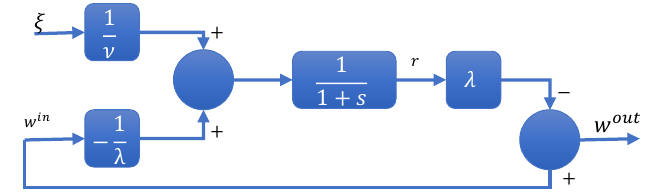}
    	\caption{Water pipeline.}
    	\label{waterpipe}
    \end{subfigure}
    \end{subfigure}
    \vspace{-2mm}
	\caption{Block diagram of CIs.}	\vspace{-12mm}
\end{figure}
\begin{table}[t!]
	\centering
	\begin{tabularx}{\columnwidth}{|c|X|c|X|}
		\hline Parameter & Description &	Parameter & Description\\\hline
		$ \delta^o_i $ & the $ i $-th generator's constant operational power angle in rad &$ T_{t_i} $& direct-axis transient time constant of mechanical power of $ i $-th generator\\\hline
			$ J_i , D_i $ & rotor inertia and damping factor of generator $ i $  &
		$ P_{m_i} $ & mechanical input power of the $ i $-th generator\\\hline
		$ P_{e_i} $ & active power demand in the subsystem $ i $ & $ \omega_i $ & angular frequency of the $ i $-th generator, in rad/s
		\\\hline
		$ x_{l_{ij}} $ & equivalent reactance of the bus between generators $ i $ and $ j $.&
		$ P_{ij} $&  Power flow in the line between generators $ i $ and $ j $\\\hline
		$ P^o_{ij} $ & linearized power flow in the constant voltage&
		$ E_i $  & constant voltage of generator $ i $\\\hline
		$ \boldsymbol{e}^g,\boldsymbol{e}^w $ & vector of all of the generator indexes supplied by natural gas and water&
		$ \theta $ & friction coefficient of water\\\hline
		$ n_{eg} $ & number of generators supplied by gas &
		$ n_{ew} $ & number of generators require water\\\hline
		$ \eta_i^g $ & the efficiency of transformation of gas fuel to mechanical power in the generator $ i $&
		$ n_l $ & number of lines connected between the subsystems\\\hline
		$ g^\textrm{in}_{ij} $ & the inlet gas pressure at connection of junction $ i $ to $ j $&
		$ g^\textrm{out}_{ij} $ &the outlet gas pressure at the connection of junction $ i $ to $ j $\\\hline
		$ g^\textrm{comp}_i $ & the pressure produced by the compressor in the junction $ i $&
		$ g^\textrm{def}_i $ & the predefined pressure in the junction $ i $\\\hline
		$\vartheta^\textrm{d}_i $ & the gas flow demand in junction $ i $&$ w^\textrm{out}_i $ & the outlet pressure of water at the pipeline $ i $\\\hline
		$ \vartheta^\textrm{out}_{ij} $ &the outlet gas flow rate at the connection of junction $ i $ to $ j $&
		$ \gamma_{ij} $ & The sector area of a pipeline between the junction $ i $ and $ j $ \\\hline
		$ n_i^\textrm{out} $ & Number of pipelines suppling junction $ i $ &
		$ n_g^\textrm{pipe} $ & Total number of the pipelines in the gas system\\\hline
		$ \eta_i^w $ &  a constant value based on a turbine's inner characteristics & 
		$ \vartheta^\textrm{in}_{ij} $ &the inlet gas flow rate at connection of junction $ i $ to $ j $
		\\\hline
		$ w^\textrm{in}_i $ & the inlet pressure of water at the pipeline $ i $&
		$ \xi_i $ & the flow rate of the water in pipeline $ i $ \\\hline
		$ r_i $ & friction of the water at pipeline $ i $ &
		$ \nu $ & the kinematic viscosity of water \\\hline
	\end{tabularx}\vspace{-4mm}
	\caption{ICI parameter description.}
	\label{tab_elecparam}\vspace{-12mm}
\end{table}
All the power infrastructure parameters are defined in Table \ref{tab_elecparam}. Using \eqref{eq:omega}, \eqref{eq:mechpower}, and \eqref{eq:line} we can summarize each power subsystem's state-space model as:
\begin{align}\label{elec:DE}\small
\dot{\boldsymbol{x}}^e_i(t)=\boldsymbol{A}^e_i\boldsymbol{x}^e_i(t)+\boldsymbol{B}^e_i\boldsymbol{u}^e_i(t)+\boldsymbol{H}^e_{i}\boldsymbol{x}^l, \,\, i=1,\dots,n_e
\end{align}\normalsize
where
$\small
\boldsymbol{x}^e_i=\left[\begin{array}{c}
x^e_{i1}\\
x^e_{i2}
\end{array}\right]=\left[\begin{array}{c}
\omega_i \\
P_{m_i}
\end{array}
\right],
\boldsymbol{u}^e_i = \left[
\begin{array}{c}
u^e_{i1}\\
\end{array}
\right]= \left[
\begin{array}{c}
P_{e_i}\\
\end{array}
\right],
\boldsymbol{A}^e_i=\left[\begin{array}{c : c }
-\frac{D_i}{J_i} & \frac{1}{J_i}\\\hdashline[2pt/2pt]\frac{-1}{P_i^oT_{t_i}} &\frac{-1}{T_{t_i}}
\end{array}\right],
\boldsymbol{B}^e_i=\left[
\begin{array}{c}
\frac{-1}{J_i}\\ 0
\end{array}
\right],
$\normalsize and $ \boldsymbol{x}^l$ is a vector with size $ n_l $ consisting all the line powers $ P_{ij} $. Also, $ \boldsymbol{H}_i^e $ is a row vector with entities:\small
\begin{align}
\boldsymbol{H}^e_{i}[m]=\begin{cases}
\frac{-1}{J_i}, & \boldsymbol{x}^l[m]=P_{l_{ij}} \textrm{ for } j=1,\dots,n_e , j\neq i \\
\frac{1}{J_i}, & \boldsymbol{x}^l[m]=P_{l_{ji}} \textrm{ for } j=1,\dots,n_e , j\neq i\\
0, & \textrm{otherwise},
\end{cases}
\end{align}\normalsize
where $ \boldsymbol{x}^l[m] $ is the $ m $-th entity of vector $ \boldsymbol{x}^l $. In \eqref{elec:DE}, the input to a subsystem, $ \boldsymbol{u}^e_i $, is the demand power from each generator, $ P_{e_i}  $. Any change in power demand results in the deviation of state variables. Since one of the state variables in \eqref{elec:DE} is the mechanical power input to the generator, $ P_{m_i} $, then any change in the demand power will result in the deviation of mechanical power input to the generator. However mechanical input to the generator is excited by an external input such as nuclear energy, coal, natural gas, wind, or water flow. Here, we focus on dynamics of the generators excited by the natural gas and for the other types of generators we assume a known mechanical input without any dynamics. We also assume that the mechanical input to the generators supplied by natural gas is proportional to the flow rate of arriving gas to the generator from the gas CI. Therefore, we have $ P_{m_i}=\eta_i^g \vartheta^\textrm{out}_i $,
where $ \eta_i^g $ and $ \vartheta^\textrm{out}_i $ are defined in Table \ref{tab_elecparam}.

Power generation also depends on some form of water input. For instance, thermal power generation requires large volumes of water for cooling purposes. Hydroelectric power requires flowing water to drive generating turbines \cite{lubega2014quantitative}. We model the dependence of the power generation on water using the boiler turbine dynamics in \cite{maffezzoni1997boiler}. Since the generator's temperature control requires water flow to the generator, the relationship between the input water and the generator's mechanical power can be written as:
$ P_{m_i}=\eta_i^w \xi_i^\textrm{out} $, where $ \eta_i^W $ and $ \xi_i^\textrm{out} $ are defined in Table \ref{tab_elecparam}. 

\eqref{elec:DE} considers each generator separately. However, to analyze the power infrastructure as a unified system we need to define a centralized model for the power system. In this regard, we present the state space model {as follows:\small
\begin{equation} \label{eq_powerCI}
\begin{aligned}
\dot{\boldsymbol{x}}^e(t)&=\boldsymbol{A}^e\boldsymbol{x}^e(t)+\boldsymbol{B}^e\boldsymbol{u}^e(t),\\
\boldsymbol{y}^{eg}(t)&=\boldsymbol{C}^{eg}\boldsymbol{x}^e(t),
\boldsymbol{y}^{ew}(t)=\boldsymbol{C}^{ew}\boldsymbol{x}^e(t),
\end{aligned}
\end{equation}}\normalsize
\small
\begin{align*}
\boldsymbol{x}^e=\left[ 
\begin{array}{c}
\boldsymbol{x}^e_1 \\
\vdots \\
\boldsymbol{x}^e_{n_e}\\
\boldsymbol{x}^l
\end{array}
\right],
\boldsymbol{u}^e=\left[
\begin{array}{c}
\boldsymbol{u}_1 \\
\vdots\\
\boldsymbol{u}_{n_e}
\end{array}
\right],
\boldsymbol{A}^e_{\left(2n_e+n_l\right)\times \left(2n_e+n_l\right)}=\left[
\begin{array}{c c c : c}
\boldsymbol{A}^e_1 & \boldsymbol{0} & \boldsymbol{0 } & \boldsymbol{H}^e_1 \\
\boldsymbol{0}& \ddots & \boldsymbol{0} & \vdots \\
\boldsymbol{0} & \boldsymbol{0} & \boldsymbol{A}^e_{n_e}& \boldsymbol{H}^e_{n_g}\\
\hdashline[2pt/2pt]
& \boldsymbol{A}^{el} &  & \boldsymbol{0}
\end{array}
\right],
\end{align*}
\begin{align*}
\boldsymbol{A}^{el}_{n_l\times 2n_e}\left[i,j\right]&=
\begin{cases}
P^o_{\frac{j}{2}m}, \hfil \textrm{if } \frac{j}{2}=1,\dots,n_e , \textrm{ \& } \boldsymbol{x}^l[i]=P_{\frac{j}{2}m},\\
-P^o_{\frac{j}{2}m}, \hfil \textrm{if } \frac{j}{2}=1,\dots,n_e , \textrm{ \& } \boldsymbol{x}^l[i]=P_{m\frac{j}{2}},\\
0, \hfil \textrm{otherwise},
\end{cases}
\end{align*}
\begin{align*}
\boldsymbol{B}^e_{{\left(2n_e+n_l\right)\times \left(n_e\right)}}=\left[
\begin{array}{c c c}
\boldsymbol{B}^e_1 &
\boldsymbol{0} &
\boldsymbol{0} \\
\boldsymbol{0} & \ddots & \boldsymbol{0} \\
\boldsymbol{0} & \boldsymbol{0} & \boldsymbol{B}^e_{n_e}\\
\boldsymbol{0} & \dots & \boldsymbol{0}
\end{array}
\right],
\boldsymbol{C}_{n_{eq}\times (2n_e+n_l)}^{eq}[i,j]=
\begin{cases}
\frac{1}{\eta_{j/2}^q} &\textrm{if } j/2=\boldsymbol{e}^q[i],\\
0, & \textrm{otherwise,}
\end{cases} \forall q \in \left\{g,w\right\}
\end{align*}\normalsize
and $ \boldsymbol{y}^{eg} $ and $ \boldsymbol{y}^{ew} $ are the vectors of natural gas and water demand from the power CI, respectively. $ \boldsymbol{e}^g $ and $ \boldsymbol{e}^w $ are the vectors which contain the indices of generators require natural gas and water.\vspace{-3mm} 
\subsection{Natural Gas Infrastructure}\vspace{-2mm} 
To analyze the performance of a natural gas CI, we need to model the transient flow of the gas pipelines, as in \cite{alamian2012state}. In this model, the inlet pressure and outlet flow rate at each pipeline are the inputs to the pipeline system. Also, the outlet gas pressure and the inlet flow rate are the outputs of the pipeline. The reason that outlet flow rate is the input of the pipeline system is that it determines the gas flow demand in the output of the system. Also, the pressure in the inlet of the pipelines is controlled by the components such as the compressor. Therefore, we consider the inlet pressure of the pipeline as another input to the pipeline system. The block diagram of a pipeline is shown in Fig. \ref{fig_gas} and the state space equations for a pipeline are derived as follows:
\begin{equation}\label{SSMgas}
\begin{aligned}
\dot{x}^g_1=\frac{-1}{\rho_1}x^g_1+\frac{\tau_1}{\rho_1}g^\textrm{in},\,\,\,
\dot{x}^g_2=\frac{-1}{\hat{\rho}_1}x^g_2 + \frac{1}{\hat{\rho}_1}g^\textrm{in},\,\,\,
\dot{x}^g_3=\frac{-1}{\hat{\rho}_1}x^g_3-\frac{\tau_2}{\hat{\rho}_1}\vartheta^\textrm{out},\,\,
\dot{x}^g_4=\frac{-1}{\tau_3}x^g_4+\frac{1}{\tau_3}\vartheta^\textrm{out},
\end{aligned}
\end{equation}
and the outputs of the pipeline are:
\begin{equation}\label{SSMgasout}
\begin{aligned}
g^\textrm{out}=x^g_1+\left(1-\frac{\hat{\tau}_1}{\hat{\rho}_1}\right)x^g_3-\frac{\hat{\tau}_1\tau_2}{\hat{\rho}_1}\vartheta^\textrm{out},
\vartheta^\textrm{in}=\frac{-\hat{\tau}_2}{\hat{\rho}_1}x^g_2+x^g_4-\frac{\hat{\tau}_2}{\hat{\rho}_1}g^\textrm{in}.
\end{aligned}
\end{equation}
The parameters $ \tau_1,\tau_2,\tau_3,\hat{\tau}_1, \rho_1$, and $ \hat{\rho}_1 $ are computed from the linearization of the pressure and gas flow rate inside the pipeline. For detailed analysis, the reader can refer to \cite[Annex A]{alamian2012state}. Other variables in the natural gas network are described in Table \ref{tab_elecparam}. 
In the natural gas CI, we need to account for the connection between the pipelines. To do so, first we consider that the summation of inlet gas flow in the junctions has to be equal to the summation of the outlet gas flow. Therefore, the inlet-outlet gas rate relationship in each junction $ i $ is $
\sum_{\left\{m|m\rightarrow i \right\}}\vartheta^\textrm{out}_{mi}=\sum_{\left\{j|i\rightarrow j \right\}}\vartheta^\textrm{in}_{ij}+\vartheta^\textrm{d}_i,$ where $
\vartheta^\textrm{out}_{mi}=\frac{\gamma_{{mi}_1}}{\sum_{\left\{q|q\rightarrow i \right\}}\gamma_{{qi}_1}}\left(\sum_{\left\{j|i\rightarrow j \right\}}\vartheta^\textrm{in}_{ij}+\vartheta^\textrm{d}_i\right),
$
and the notation $ j|j\rightarrow i $ means all the junctions supplying junction $ i $ and $ j|i\rightarrow j $ means all the junctions supplied by junction $ i $. Also, the inlet-outlet pressure relationship in each junction can be written as $
g_i^\textrm{out}\triangleq\frac{1}{n_i^\textrm{out}}\sum_{j|j\rightarrow i}g_{ji}^\textrm{out},\Rightarrow
g_i^\textrm{out}+g_i^\text{comp}=g_i^\textrm{def},\Rightarrow
g_{ij}^\textrm{in}=g_i^\textrm{def},  \forall j|i\rightarrow j.$
A compressor at each junction compensates the pressure loss \cite{pambour2016integrated}. The required power at each compressor $ i $ is a function of the pressure produced by the compressor as$P_{e_i}=\eta_i^{c}(g_i^\textrm{def}-g_i^\textrm{out}),$
where $ g_i^\textrm{def} $ is defined by the designer of the system based on the characteristics of the junction. Therefore, the gas pressure in the inlet of the pipelines is fixed. Now, we can write the state space model of a natural gas pipeline between junction $ i $ and $ j $ as:\small
\begin{align}
\dot{\boldsymbol{x}}^g_{ij}&=\boldsymbol{A}^g_{ij}\boldsymbol{x}^g_{ij}+\boldsymbol{B}^g_{ij}\left(\boldsymbol{u}^g_{ij}+\sum_{m|j\rightarrow m}^{n}\boldsymbol{y}^g_{jm}\right), \hfil i=1,\dots,n_g\nonumber\\\label{gas:DE}
\boldsymbol{y}^g_{ij}&=\boldsymbol{C}^g_{ij}\boldsymbol{x}_{ij}^g+\boldsymbol{D}^g_{ij}\left(\boldsymbol{u}^g_{ij}+\sum_{m|j\rightarrow m}^{n}\boldsymbol{y}^g_{jm}\right),
\boldsymbol{y}^{ge}_{i}=\boldsymbol{C}^{ge}_{i}\left(\boldsymbol{u}^{ge}_{i}-\boldsymbol{D}^{ge}_{i}\sum_{j|j\rightarrow i}^{n}\boldsymbol{y}^g_{ji}\right),
\end{align}\normalsize
where\small
\begin{align*}
\boldsymbol{x}^g_{ij}&=\left[\begin{array}{c}
x^g_{{ij}_1}\\
x^g_{{ij}_2}\\
x^g_{{ij}_3}\\
x^g_{{ij}_4}
\end{array}\right],\boldsymbol{u}^g_{ij} = \left[
\begin{array}{c}
u^g_{{ij}_1}\\
u^g_{{ij}_2}
\end{array}
\right]= \left[
\begin{array}{c}
g_i^\textrm{def}\\
\vartheta_j^\textrm{d}
\end{array}
\right],
\boldsymbol{y}^g_{ij}=\left[\begin{array}{c}
y^g_{{ij}_1}\\
y^g_{{ij}_2}
\end{array}\right]=\left[\begin{array}{c}
g_{ij}^\textrm{out}\\
\vartheta_{{ij}^\textrm{in}}
\end{array}\right],\boldsymbol{y}^{ge}_i=\left[\begin{array}{c}
P_{e_i}\\
\end{array}\right],\\\boldsymbol{C}^g_{ij}&=\left[
\begin{array}{c  c  c  c}
1 & 0 & \left(1-\frac{\hat{\tau}_{{ij}_1}}{\hat{\rho}_{{ij}_1}}\right) & 0\\
0 & \frac{-\hat{\tau}_{{ij}_2}}{\hat{\rho}_{{ij}_1}} &  0 & 1
\end{array}
\right],
\boldsymbol{C}_i^{ge}=\left[
\begin{array}{c}
\eta_i^c
\end{array}
\right],\boldsymbol{D}_i^{ge}=\left[
\begin{array}{c c}
\frac{1}{n_i^\textrm{out}} & 0
\end{array}
\right],\boldsymbol{u}_i^{ge}=\left[
\begin{array}{c}
g_i^\textrm{def}
\end{array}
\right],\\
\boldsymbol{A}^g_{ij}&=\left[\begin{array}{c  c  c  c}
\frac{-1}{\rho_{{ij}_1}}& 1 & 0 & 0\\
0 & \frac{-1}{\hat{\rho}_{{ij}_1}}& 0 & 0 \\
0 & 0 & \frac{-1}{\hat{\rho}_{{ij}_1} }& 0 \\ 
0 & 0 & 0 & \frac{-1}{\tau_{{ij}_1}}
\end{array}\right],
\boldsymbol{B}^g_{ij}=\left[
\begin{array}{c c}
\frac{\tau_{i1}}{\rho_{i1}} & 0 \\
\frac{1}{\hat{\rho}_{i1}} & 0\\
0& \frac{-\tau^2_{{ji}_1}}{\hat{\rho}_{{ij}_1}\sum_{\left\{j|j\rightarrow i \right\}}\tau_{{ji}_1}}\\
0 & \frac{\tau_{{ji}_1}}{\tau_{{ij}_3}\sum_{\left\{j|j\rightarrow i \right\}}\tau_{{ji}_1}}
\end{array}
\right]\hspace{-1mm},\boldsymbol{D}^g_{ij}\hspace{-1mm}=\hspace{-1mm}\left[
\begin{array}{c  c}
0 & \frac{\hat{\tau}_{{ij}_1}\tau_{{ij}_2}}{\hat{\rho}_{{ij}_1}}\\
\frac{\hat{\tau}_{{ij}_1}}{\hat{\rho}_{{ij}_1}}& 0
\end{array}
\right]\hspace{-1mm}.
\end{align*}\normalsize
and $ \boldsymbol{y}_i^{ge} $ is defined the power demand in the compressor of each junction $ i $. The relationship in \eqref{gas:DE} considers the dynamics of a single pipeline system. To analyze the natural gas CI as a unified system, we propose the centralized model consisting all of the state variables {as:\small
\begin{equation}\label{eq_gasCI}
\begin{aligned}
\dot{\boldsymbol{x}}^g(t)=\boldsymbol{A}^g\boldsymbol{x}^g(t)+\boldsymbol{B}^g\boldsymbol{u}^g(t),\,\,
\boldsymbol{y}^{ge}(t)=\boldsymbol{C}^{ge}\boldsymbol{x}^g(t)+\boldsymbol{D}^g\boldsymbol{u}^g(t),
\end{aligned} 
\end{equation}}\normalsize
 where\small
\begin{align}
\label{gasstatevec}
\boldsymbol{x}^{g}\hspace{-1mm}=\hspace{-1mm}\left[
\begin{array}{c}
\boldsymbol{x}^{g^T}_1,
\cdots,
\boldsymbol{x}^{g^T}_{n_g^\textrm{pipe}}
\end{array}
\right]^T,	\boldsymbol{u}^{g}\hspace{-1mm}=\hspace{-1mm} \left[
\begin{array}{c}
\boldsymbol{u}^{g^T}_1,
\cdots,
\boldsymbol{u}^{g^T}_{n_g^\textrm{pipe}}
\end{array}
\right]^T,\boldsymbol{y}^{ge}\hspace{-1mm}=\hspace{-1mm}\left[
\begin{array}{c}
\boldsymbol{y}^{ge^T}_1 ,
\cdots ,
\boldsymbol{y}^{ge^T}_{n_g}
\end{array}
\right]^T,
\end{align}\normalsize
{and $ \boldsymbol{A}^g $, $ \boldsymbol{B}^g $, $ \boldsymbol{C}^g $, and $ \boldsymbol{D}^g $ are $ 4n_g^{\textrm{pip}}\times 4n_g^{\textrm{pip}} $, $ 4n_g^{\textrm{pip}}\times 2n_g^{\textrm{pip}} $, $ 2n_g^{\textrm{pip}}\times4n_g^{\textrm{pip}} $, and $ 2n_g^{\textrm{pip}}\times2n_g^{\textrm{pip}}  $ matrices, respectively. }
In \eqref{gasstatevec}, each vector $ \boldsymbol{x}^g_m $ corresponds to one of the pipeline state vectors, $ \boldsymbol{x}^g_{ij}$. Moreover, each vector $ \boldsymbol{u}^g_m $ in \eqref{gasstatevec} corresponds to one of the pipeline input vectors, $ \boldsymbol{u}^g_{ij}$. To find matrices $ \boldsymbol{A}^g$, $ \boldsymbol{B}^g $, $ \boldsymbol{C}^g $, $ \boldsymbol{D}^g $, we need to start from the last pipelines in the natural gas CI and find the outputs of these pipelines and use them as the inputs of previous pipelines. By proceeding this method until reaching the source junction we can derive the mentioned matrices.\vspace{-5mm}
\subsection{Water Infrastructure}\vspace{-3mm}
To analyze the water CI we use the model presented in \cite{Burgschweiger2009}, in which the flow rate is considered constant in the outlet and inlet of each pipeline. Moreover, the relationship between the flow, inlet and outlet pressure of the pipeline can be expressed as $
\dot{r}=-r+\frac{1}{\nu}\xi-\frac{1}{\theta}w^\textrm{in},\,\,
w^\textrm{out}=-\theta r+w^\textrm{in}, $
where all the parameters are described in Table \ref{tab_elecparam}. Fig. \ref{waterpipe} shows the block diagram model of the pipeline in water system. Also at each junction of water system the summation of inlet flow has to be equal to the summation of the outlet flow of water. Therefore, at each junction $ i $ we have:
\begin{equation}\label{waterflow}\small
\begin{aligned}
\sum_{\left\{m|m\rightarrow i \right\}}\xi_{mi}=\sum_{\left\{j|i\rightarrow j \right\}}\xi_{ij}+\xi^\textrm{d}_i,\Rightarrow
\xi_{mi}=\frac{\gamma_{mi}}{\sum_{\left\{q|q\rightarrow i \right\}}\gamma_{qi}}\left(\sum_{\left\{j|i\rightarrow j \right\}}\xi_{ij}+\xi^\textrm{d}_i\right).
\end{aligned}
\end{equation}
Moreover, due to the pressure loss in the pipelines, there exist a water pump in each junction which controls the pressure of water in the junctions. Then, the water pressure relationship in each junction is $
w_i^\textrm{out}\triangleq\frac{1}{n_i^\textrm{out}}\sum_{j|j\rightarrow i}w_{ji}^\textrm{out},\Rightarrow\nonumber
w_i^\textrm{out}+w_i^\text{pump}=w_i^\textrm{def},\label{waterpressure}\Rightarrow
w_{ij}^\textrm{in}=w_i^\textrm{def}, \quad \forall j|i\rightarrow j.$
Each junction $ i $'s required power for the pump is a function of pressure provided by the pump as $
P_{e_i}=\eta_i^{p}(w_i^\textrm{def}-w_i^\textrm{out}),$ where $ w_i^\textrm{def} $ is a predefined value by the designer of the water system which indicates the fixed inlet water pressure in each pipe. Similar to the natural gas CI, the state space model for each water pipeline is as follows:\small
\begin{align}
\dot{\boldsymbol{x}}^w_{ij}&=\boldsymbol{A}^w_{ij}\boldsymbol{x}^w_{ij}+\boldsymbol{B}^w_{ij}\left(\boldsymbol{u}^w_{ij}+\sum_{m|j\rightarrow m}^{n}\boldsymbol{y}^w_{jm}\right), \hfil i=1,\dots,n_w\nonumber\\\label{water:DE}
\boldsymbol{y}^w_{ij}&=\boldsymbol{C}^w_{ij}\boldsymbol{x}_{ij}^w+\boldsymbol{D}^w_{ij}\left(\boldsymbol{u}^w_{ij}+\sum_{m|j\rightarrow m}^{n}\boldsymbol{y}^w_{jm}\right), \quad
\boldsymbol{y}^{we}_{i}=\boldsymbol{C}^{we}_{i}\left(\boldsymbol{u}^{we}_{i}-\boldsymbol{D}^{we}_{i}\sum_{j|j\rightarrow i}^{n}\boldsymbol{y}^w_{ji}\right),
\end{align}\normalsize
where\small
\begin{align*}
\boldsymbol{x}^w_{ij}&=\left[\begin{array}{c}
r_{ij}
\end{array}\right],
\boldsymbol{u}^w_{ij} = \left[
\begin{array}{c}
w_i^\textrm{def}\\
\xi_{j}^d
\end{array}
\right],
\boldsymbol{y}^w_{ij}=\left[\begin{array}{c}
r_{ij}
\end{array}\right],
\boldsymbol{A}^w_{ij}=\left[\begin{array}{c}
-1
\end{array}\right],\boldsymbol{B}^g_{ij}=\left[
\begin{array}{c c}
\frac{1}{\nu_{ij}} & \frac{-1}{\theta_{ij}} \\
\end{array}
\right],\\
\boldsymbol{C}^w_{ij}\hspace{-1mm}&=\hspace{-1mm}\left[
\begin{array}{c}
-\theta_{ij}\\
0
\end{array}
\right],
\boldsymbol{D}^w_{ij}\hspace{-1mm}=\hspace{-1mm}\left[
\begin{array}{c c}
1 & 0\\
0 & 0
\end{array}
\right],
\boldsymbol{y}^{we}_i\hspace{-1mm}=\hspace{-1mm}\left[\begin{array}{c}
P_{e_i}\\
\end{array}\right],
\boldsymbol{C}_i^{we}\hspace{-1mm}=\hspace{-1mm}\left[
\begin{array}{c}
\eta_i^p
\end{array}
\right],
\boldsymbol{D}_i^{we}\hspace{-1mm}=\hspace{-1mm}\left[
\begin{array}{c c}
\frac{1}{n_i^\textrm{out}} & 0
\end{array}
\right],
\boldsymbol{u}_i^{we}\hspace{-1mm}=\hspace{-1mm}\left[
\begin{array}{c}
w_i^\textrm{def}
\end{array}
\right].
\end{align*}\normalsize
Vector $ \boldsymbol{y}_i^{we} $ identifies the power demand from the pumps in each junction. The system presented in \eqref{water:DE} is a decentralized model which considers only one pipeline. Using \eqref{waterflow} and \eqref{water:DE}, for each pipeline in the water system we can find the centralized water CI model as follows:
{\begin{equation}\label{eq_waterCI}
\begin{aligned}
\dot{\boldsymbol{x}}^w(t)=\boldsymbol{A}^w\boldsymbol{x}^w(t)+\boldsymbol{B}^w\boldsymbol{u}^w(t),\quad
\boldsymbol{y}^{we}(t)=\boldsymbol{C}^{we}\boldsymbol{x}^w(t)+\boldsymbol{D}^w\boldsymbol{u}^w(t),
\end{aligned} 
\end{equation}} 
where \small $
\boldsymbol{x}^w\hspace{-1mm}=\hspace{-1mm}\left[
\begin{array}{c}
\boldsymbol{x}^w_1,
\cdots,
\boldsymbol{x}^w_{n_w^\textrm{pipe}}
\end{array}
\right]^T, 
\boldsymbol{u}^w\hspace{-1mm} = \hspace{-1mm}\left[
\begin{array}{c}
\boldsymbol{u}^w_1,
\cdots,
\boldsymbol{u}^w_{n_w^\textrm{pipe}}
\end{array}
\right]^T,
\boldsymbol{y}^{we}\hspace{-1mm}=\hspace{-1mm}\left[
\begin{array}{c}
\boldsymbol{y}^{we}_1 ,
\cdots ,
\boldsymbol{y}^{we}_{n_w^{\textrm{pipe}}}
\end{array}
\right]^T,$ \normalsize
and each $ \boldsymbol{x}^w_m $ corresponds to one of the pipeline state vectors, $ \boldsymbol{x}^w_{ij}$, and each $ \boldsymbol{u}^w_m $ corresponds to one of the pipeline input vectors, $ \boldsymbol{u}^w_{ij}$. {Also, $ \boldsymbol{A}^w $, $ \boldsymbol{B}^w $, $ \boldsymbol{C}^w $, and $ \boldsymbol{D}^w $ are $ n_w^{\textrm{pipe}}\times n_w^{\textrm{pipe}} $, $ n_w^{\textrm{pipe}}\times 2n_w^{\textrm{pip}} $, $ 2n_w^{\textrm{pipe}}\times n_w^{\textrm{pipe}} $, and $ 2n_w^{\textrm{pipe}}\times 2n_w^{\textrm{pipe}} $ matrices, respectively,} and to find these matrices the same procedure in the natural gas CI model can be proceeded. We start from the junctions that only have the demand outlet flow without any pipelines connected to their outlet. Then, we use each pipeline's output as inputs to the pipelines of previous step. We continue this procedure until we reach the source.\vspace{-5mm}
\subsection{ICI Model}\vspace{-3mm}
{
Considering the power-gas-water CI interdependence, a unified model for each CI is:\small
\begin{align}
\dot{\boldsymbol{x}}^e(t)&=\boldsymbol{A}^e\boldsymbol{x}^e(t)+\boldsymbol{B}^e\left(\boldsymbol{u}^e(t)+\boldsymbol{T}^{ge}\boldsymbol{y}^{ge}(t)+\boldsymbol{T}^{we}\boldsymbol{y}^{we}(t)\right),\nonumber\\
\dot{\boldsymbol{x}}^g(t)&=\boldsymbol{A}^g\boldsymbol{x}^g(t)+\boldsymbol{B}^g\left(\boldsymbol{u}^g(t)+\boldsymbol{T}^{eg}\boldsymbol{y}^{eg}(t)\right),\label{interCIstate}\quad
\dot{\boldsymbol{x}}^w(t)=\boldsymbol{A}^w\boldsymbol{x}^w(t)+\boldsymbol{B}^w\left(\boldsymbol{u}^w(t)+\boldsymbol{T}^{ew}\boldsymbol{y}^{ew}(t)\right),
\end{align}\normalsize
where \small
\begin{align}\label{interCIout}
\begin{aligned}
\boldsymbol{y}^{eg}(t)&=\boldsymbol{C}^{eg}\boldsymbol{x}^e(t),
\boldsymbol{y}^{ew}(t)=\boldsymbol{C}^{ew}\boldsymbol{x}^e(t),\\
\boldsymbol{y}^{ge}(t)&=\boldsymbol{C}^{ge}\boldsymbol{x}^g(t)+\boldsymbol{D}^g\left(\boldsymbol{u}^g(t)+\boldsymbol{T}^{eg}\boldsymbol{y}^{eg}(t)\right),\,\,
\boldsymbol{y}^{we}(t)=\boldsymbol{C}^{we}\boldsymbol{x}^w(t)+\boldsymbol{D}^w\left(\boldsymbol{u}^w(t)+\boldsymbol{T}^{ew}\boldsymbol{y}^{ew}(t)\right).
\end{aligned}
\end{align}\normalsize
Here, $ \boldsymbol{T}^{ge} $, $ \boldsymbol{T}^{we} $, $ \boldsymbol{T}^{ew} $, and $ \boldsymbol{T}^{eg} $ are $ 4n_g^{\textrm{pipe}}\times (2n_e+n_l) $, $ n_w^{\textrm{pipe}}\times (2n_e+n_l) $, $ {(2n_e+n_l)\times n_w^\textrm{pipe}} $, and $ {(2n_e+n_l)\times 4n_g^\textrm{pipe}} $ matrices connecting the inputs and outputs of the three CIs whose elements are equal to one if the output of one CI is connected to the input of another CI or is equal to zero otherwise. By substituting \eqref{interCIout} into \eqref{interCIstate}, we will have the following state-space model for the interdependent critical gas-power-water infrastructure as in \eqref{eq:interCI},
where
\begin{align}\label{eq:States}
\boldsymbol{x}(t)&=\left[
\begin{array}{c}
\boldsymbol{x}^{e^T}(t),
\boldsymbol{x}^{g^T}(t),
\boldsymbol{x}^{w^T}(t)
\end{array}
\right]^T,\boldsymbol{u}(t)=\left[
\begin{array}{c}
\boldsymbol{u}^{e^T}(t),
\boldsymbol{u}^{g^T}(t),
\boldsymbol{u}^{w^T}(t)
\end{array}
\right]^T,
\end{align} 
and $ \bar{\boldsymbol{A}} $ and $ \bar{\boldsymbol{B}} $ are $ n\times n $ and $ {n\times \tilde{n}} $ matrices defined in \eqref{matrixAdef} where $ n\triangleq 2n_e+n_l+4n_g^{\textrm{pipe}}+n_w^{\textrm{pipe}} $, and $ \tilde{n} \triangleq n_e + 2n_g^{\textrm{pipe}} + 2n_w^{\textrm{pipe}}$. \eqref{eq:interCI} captures the dynamics of an ICI. In this model, the state variables of three CIs are mutually interdependent, and changes in one CI can affect the other two CIs. }

{\small
	\begin{equation}\label{matrixAdef}
	\scriptstyle\bar{\boldsymbol{A}}\triangleq \left[
	\begin{array}{c : c : c}
	{\scriptstyle\boldsymbol{A}^e+\boldsymbol{B}^e\left(\boldsymbol{T}^{ge}\boldsymbol{D}^g\boldsymbol{T}^{eg}\boldsymbol{C}^{eg}+\boldsymbol{T}^{we}\boldsymbol{D}^w\boldsymbol{T}^{ew}\boldsymbol{C}^{ew}\right)} &{\scriptstyle \boldsymbol{B}^e\boldsymbol{T}^{ge}\boldsymbol{C}^{ge} }&{\scriptstyle \boldsymbol{B}^e\boldsymbol{T}^{we}\boldsymbol{C}^{we}}\\\hdashline[2pt/2pt]
	{\scriptstyle\boldsymbol{B}^g\boldsymbol{T}^{eg}\boldsymbol{C}^{eg}}&{\scriptstyle \boldsymbol{A}^g}&{\scriptstyle\boldsymbol{0}}\\\hdashline[2pt/2pt]
	{\scriptstyle\boldsymbol{B}^w\boldsymbol{T}^{ew}\boldsymbol{C}^{ew}}&{\scriptstyle\boldsymbol{0}}& {\scriptstyle\boldsymbol{A}^w}
	\end{array}
	\right],
	\scriptstyle\bar{\boldsymbol{B}}\triangleq \left[
	\begin{array}{c : c : c}
	\scriptstyle\boldsymbol{B}^e &\scriptstyle \boldsymbol{B}^e\boldsymbol{T}^{ge}\boldsymbol{D}^g &\scriptstyle \boldsymbol{B}^e\boldsymbol{T}^{we}\boldsymbol{D}^w\\\hdashline[2pt/2pt]\scriptstyle
	\boldsymbol{0}&\scriptstyle\boldsymbol{B}^g&\scriptstyle\boldsymbol{0}\\\hdashline[2pt/2pt]\scriptstyle
	\boldsymbol{0}&\scriptstyle\boldsymbol{0}& \scriptstyle\boldsymbol{B}^w
	\end{array}
	\right].
	\end{equation}
	\vspace{-5mm}}

In summary, we proposed the state space modeling for each of the CIs and analyzed their interdependence. Moreover, we derived the matrices of dynamic model of each CI.\vspace{-4mm}}
\end{appendices}
\def\baselinestretch{0.83}
\bibliographystyle{IEEEtran}
\bibliography{references}\vspace{-3mm}
\end{document}